\documentclass[10pt, english, a4paper]{article}
\usepackage[T1]{url}
\usepackage[utf8]{inputenc}
\usepackage[hidelinks]{hyperref} 
\usepackage{graphicx}
\usepackage{varioref, babel, fancyvrb, listings, amsmath,amsthm, amssymb, enumerate, mathtools}

\usepackage[left=1.25in,right=1.25in,top=1.25in,bottom=1.25in]{geometry} 
\usepackage{benstyle_simple}
\usepackage{mathabx}
\usepackage{listings}
\usepackage{todonotes}
\usepackage{subcaption}
\usepackage{array}
\usepackage{csquotes}
\usepackage{tikz}
\usepackage[normalem]{ulem}

\usetikzlibrary{patterns}

\theoremstyle{plain}
\newtheorem{theorem}{Theorem}[section]
\newtheorem{lemma}[theorem]{Lemma}
\newtheorem{proposition}[theorem]{Proposition}
\newtheorem{corollary}[theorem]{Corollary}

\theoremstyle{definition}
\newtheorem{definition}[theorem]{Definition}

\theoremstyle{remark}
\newtheorem{remark}[theorem]{Remark}

\def\m{\boldsymbol{m}}

\def\s{\boldsymbol{s}}
\def\t{\boldsymbol{t}}

\def\M{\mathbf{M}}
\def\Nm{\mathbf{N}} 

\newcommand{\Bs}{\mathcal{B}}

\newcommand{\Fs}{\mathcal{F}}

\newcommand{\Ws}{\mathcal{W}}

\def\C{\mathbb{C}}
\def\Z{\mathbb{Z}}
\def\N{\mathbb{N}}

\def\R{\mathbb{R}}

\def\d{~\text{d}}

\def\supp{\textnormal{supp}}

\def\({\left(}
\def\){\right)}

\def\eps{\epsilon}

\def\La{\Lambda}
\def\om{\omega}
\def\Om{\Omega}
\def\al{\alpha}

\def\bs{\boldsymbol}

\newcommand{\ind}[1]{\left\langle #1 \right\rangle}
\newcommand{\indi}[1]{\langle #1 \rangle} 
\newcommand{\floor}[1]{\left\lfloor #1 \right\rfloor}
\newcommand{\ceil}[1]{\left\lceil #1 \right\rceil}

\DeclareMathOperator{\Supp}{supp} 
\DeclareMathOperator*{\minimize}{minimize}

\def\Bwh{B_{\textnormal{wh}}}

\def\Bwave{B_{\textnormal{wave}}}
\def\Bsa{B_{\textnormal{sa}}}
\def\Bsp{B_{\textnormal{sp}}}
\def\bsa{b^{\textnormal{sa}}}
\def\bsp{b^{\textnormal{sp}}}
\def\Sboms{S_{\bom,\s}}
\def\zetas{\zeta_{\s,\bom}}

\def\Llog{r \cdot \log(2\tilde{m})\cdot \log(2N)\cdot\log^2(2s) + \log(\eps^{-1})}
\def\Lanul{\Lambda_{\nu, j, \textnormal{left}}}
\def\Lanum{\Lambda_{\nu, j, \textnormal{mid}}}
\def\Lanur{\Lambda_{\nu, j, \textnormal{right}}}

\def\bom{\boldsymbol{\omega}}

\title{Uniform recovery in infinite-dimensional compressed sensing and
applications to structured binary sampling}

\author{Ben Adcock \footnotemark[1]  
\and
Vegard Antun \footnotemark[2] \footnotemark[4]
\and Anders C. Hansen \footnotemark[3] \footnotemark[2]
}

\begin{document}

\maketitle

\footnotetext[1]{Simon Fraser University, Canada}
\footnotetext[2]{University of Oslo, Norway}
\footnotetext[3]{University of Cambridge, United Kingdom}
\footnotetext[4]{Corresponding author (vegarant@math.uio.no)}

\begin{abstract}
Infinite-dimensional compressed sensing deals with the recovery of analog
signals (functions) from linear measurements, often in the form of integral
transforms such as the Fourier transform.  This framework is well-suited to
many real-world inverse problems, which are typically modelled in
infinite-dimensional spaces, and where the application of finite-dimensional
approaches can lead to noticeable artefacts.  Another typical feature of such
problems is that the signals are not only sparse in some dictionary, but
possess a so-called local sparsity in levels structure.  Consequently, the
sampling scheme should be designed so as to exploit this additional structure.
In this paper, we introduce a series of uniform recovery guarantees for
infinite-dimensional compressed sensing based on sparsity in levels and
so-called multilevel random subsampling.  By using a weighted
$\ell^1$-regularizer we derive measurement conditions that are sharp up to log
factors, in the sense they agree with those of certain oracle estimators.
These guarantees also apply in finite dimensions, and improve existing results
for unweighted $\ell^1$-regularization.  To illustrate our results, we consider
the problem of binary sampling with the Walsh transform using orthogonal
wavelets.  Binary sampling is an important mechanism for certain imaging
modalities.  Through carefully estimating the local coherence between the Walsh
and wavelet bases, we derive the first known recovery guarantees for this
problem.
\end{abstract}

\paragraph*{Keywords:} Infinite-dimensional compressed sensing, uniform
recovery, Walsh sampling, wavelet recovery, sparsity in levels, local coherence

\paragraph*{Mathematics Subject Classification (2010):} 94A20, 42C40, 42C10, 15B52 

\section{Introduction}\label{s:introduction}

Compressive sensing (CS), introduced by Candes, Romberg
\& Tao in \cite{Candes06} and Donoho in \cite{Donoho06}, has been an area of substantial research during the last
decade.  The key assumption, which lays
the foundation for this field of research, is that a sparse vector
$x \in \C^M$ can be recovered from an underdetermined system of linear
equations, using, for instance, convex optimization algorithms \cite{Eldar12, Foucart13}.

Imaging has been one of the most successful areas of application of CS.  However, in this area, the sparsity assumption is typically too general.  Examples include
all applications using Fourier samples -- 
such as Magnetic Resonance Imaging (MRI) \cite{Larson11, Lustig07, Lustig08},
surface scattering \cite{Jones16atom}, Computerized Tomography (CT) and
electron microscopy -- as well as applications using binary sampling, e.g.\ fluorescence microscopy \cite{Studer12}, 
lensless imaging \cite{Zomet06} and numerous other optical imaging modalities \cite{CSCodedAp,GehmBradyCSOpt,WillettCSOptTut}. 
Natural images, when sparsified via a wavelet (or more generally, $X$-let)
transform, are not only sparse, but have specific sparsity structure
\cite{Adcock17,AsymptoticCS}.
For wavelets, which will be our sparsifying transform in
this paper, natural images have coefficients where most of the large
entries are concentrated at the coarse scales, and progressively fewer at
the fine scales (termed \textit{asymptotic sparsity} in \cite{Adcock17}).

In the presence of structured sparsity, it is natural to ask how best to promote this additional structure.  In \cite{Adcock17} it was proposed to do this via the sampling operator.  Wavelets partition Fourier space into dyadic bands corresponding to distinct scales.  Hence, by choosing Fourier samples in these bands corresponding to the local sparsities, one obtains as structured sampling scheme -- a so-called \textit{multilevel sampling} scheme -- which promotes the asymptotic sparsity structure.  The practical benefits of such schemes have been demonstrated in \cite{AsymptoticCS} for various different imaging modalities, including MRI, Nuclear Magnetic Resonance (NMR) spectroscopy, fluorescence microscopy and Helium Atom Scattering.  Theoretical analysis has been presented in \cite{Adcock17} (nonuniform recovery) and \cite{Bastounis17,Li17} (uniform recovery in the finite-dimensional setting).

\subsection{Main results}
This paper has two main objectives.  First, we generalize
existing uniform recovery guarantees \cite{Bastounis17,Li17} from the
finite-dimensional to the infinite-dimensional setting.  This extension is
important for practical imaging.  Although much of the compressive imaging
literature considers the recovery of discrete images (i.e.\ finite-dimensional
arrays) from discrete measurements (e.g.\ the discrete Fourier transform),
modalities such as MRI, NMR and others are naturally analog, and hence better
modelled over the continuum (i.e.\ functions, and the continuous Fourier
transform).  Indeed, as we will see in Section \ref{subsec:shortcomings}, discretizing such a problem leads to measurement mismatch
\cite{chi11}, and in the case of wavelet recovery, the wavelet crime
\cite[232]{Strang96}, both of which can introduce artefacts in the
reconstruction \cite{Guerquin12}. In this paper, we consider signals as
functions $f \in L^2([0,1))$ and work with continuous integral transforms, thus
avoiding these pitfalls.

In our theoretical analysis, we also improve the uniform recovery guarantee given in previous works \cite{Bastounis17,Li17}.  Unlike previous results, our recovery guarantees are, up to log factors, optimal: specifically, they agree with those of the oracle least-square estimator based on \textit{a priori} knowledge of the support \cite{ABB18}.  We do this by replacing the standard $\ell^1$-minimization decoder by a certain weighted $\ell^1$-minimization decoder; an idea originally proposed in \cite{Traonmilin15}.

Our second objective is to consider binary sampling.  Previous works have addressed the case of (discrete or continuous) Fourier sampling.  Yet many imaging modalities, e.g.\ fluorescence microscopy and lensless imaging, require binary sampling operators.  To do so, we replace the Fourier transform
\[ \Fs f(\om) \coloneqq \int_{[0,1)} f(x)e^{-2\pi \om  x } \d x, ~~~~~
f \in L^{2}([0,1)),\] 
by the binary \emph{Walsh
transform}
\[ \Ws f(n) \coloneqq \int_{[0,1)} f(x) w_{n}(x) \d x, ~~~~~ f \in
L^{2}([0,1)) \]
where $w_{n}\colon [0,1)\to \{+1,-1\}$, $n \in \Z_{+} \coloneqq
\{0,1,\ldots\}$ denote the Walsh functions.
This is a widely used sampling operator in binary imaging 
\cite{Studer12, Zomet06}, and often goes under the name of Hadamard
sampling in the discrete case.  Working with this continuous transform, we
provide analogous guarantees for binary sampling to those for Fourier sampling.
As a side note, we remark that working in the continuous setting also
simplifies the analysis (specifically, the derivation of so-called
\textit{local coherence} estimates) over working directly with the discrete
setup.

We note that in this paper we only consider recovery guarantees for one
dimensional functions. We expect that the setup for higher dimensional function
will deviate slightly from what we present here, and we will save this
discussion for future work.

The outline of the remainder of this paper is as follows.  We commence in Section \ref{sec:fin_dim} by reviewing previous work, and in particular, the existing finite-dimensional theory.  We then introduce an abstract infinite-dimensional model for isometries $U$ acting on $\ell^2(\N)$ in Section \ref{sec:infin_dim}.
Here we will derive sufficient conditions for such 
operators to provide uniform recovery guarantees.  In Section \ref{sec:4} we continue this 
work by finding conditions for which the cross-Gramian $U$ between a 
wavelet and Walsh basis satisfies these conditions. Finally in Section \ref{sec:proof1}, 
\ref{sec:coher} and \ref{sec:proof2} we will present proofs of our main results.

\section{Sparsity in levels in finite dimensions}
\label{sec:fin_dim}

\subsection{Notation}
For $N\in \N$ and $\Om \subseteq \{1,\ldots, N\}$ we let $P_{\Om} \in
\C^{N\times N}$ denote the projection onto the linear span of the associated subset of the canonical basis, i.e. for 
$x\in \C^N$, we have $(P_{\Om}x)_{i} = x_i$ if $i\in \Om$ and $(P_{\Om}x)_{i} = 0$ if 
$i\not\in \Om$. Sometimes, we will abuse this notation slightly by assuming $P_{\Om} \in 
\C^{|\Om|\times N}$, and discard all the zero entries in $P_{\Om}x$. Whether we 
mean $P_{\Om}\in \C^{N\times N}$ or $P_{\Om} \in \C^{|\Om|\times N}$ will be clear 
from the context. If $\Om = \{N_{k-1}+1, \ldots, N_{k}\}$ we simply write 
$P_{N_k}^{N_{k-1}} = P_{\{N_{k-1}+1,\ldots, N_{k}\}}$, and simply $P_{N_k}$ if $N_{k-1}=0$.

We call a vector $x \in \C^N$ $s$-sparse if $|\supp(x)| \leq s$, where $\supp(x) = \{ i : x_i \neq 0 \}$. We write 
$A \lesssim B$ if there exits a constant $C> 0$ independent of all relevant
parameters, so that $A \leq C B$, and similarly for $A \gtrsim B$.

\subsection{Finite model}
Let $V \in \C^{N\times N}$ be a measurement matrix e.g.\ a Fourier of Hadamard matrix,
denoted $V_{\text{Four}}$ and
$V_{\text{Had}}$, respectively, and let $\Om \subset \{1,\ldots N\}$ with $|\Om|=m < N$.
In a typical finite-dimensional CS setup we consider the recovery of a signal $x \in \C^N$ from measurements $y = P_{\Om} V x + e \in \C^m$, where $e \in \C^m$ is a vector of
measurement error. If $x$ is sparse in a discrete wavelet basis, one then
recovers its coefficients by solving the optimization problem 
\begin{equation}
    \label{eq:QCBP_fin}
    \minimize_{z \in \C^N} \|z\|_1 \quad\text{ subject to }\quad
    \|P_{\Om}V\Psi^{-1}z - y\|_2 \leq \eta
\end{equation}
where $\Psi \in \C^{N\times N}$ is a discrete wavelet transform and
$\eta \geq \|e\|_2$ is a noise parameter. Usually one would scale $V
\in \C^{N\times N}$ so that it becomes orthonormal and choose an orthonormal
wavelet basis, so that the matrix $U = V \Psi^{-1} = V \Psi^T$
acts as an isometry on $\C^N$.   

Suppose that $U$ is indeed an isometry.  To obtain a uniform recovery guarantee for the above system, one typically first shows
that the matrix $A = \frac{1}{\sqrt{p}}P_{\Om}U \in
\C^{m\times N}$, with $p=\frac{m}{N}$, satisfies the \emph{Restricted Isometry Property} (RIP) with
high probability.  
\begin{definition}[RIP]
Let $1 \leq s \leq N$ and $A \in \C^{m \times N}$.  The \emph{Restricted Isometry Constant (RIC)} of order $s$ is the smallest $\delta \geq 0$ such that   
\[ (1-\delta)\|x\|_{2}^{2} \leq \|Ax\|_{2}^{2} \leq (1+\delta)\|x\|_{2}^{2} 
\quad \forall x \in \Sigma_{s}, \]
where $\Sigma_s$ denotes the set of $s$-sparse vectors in $\C^N$.  If $0 \leq \delta <1$ we say that $A$ has the \emph{Restricted Isometry Property (RIP)} of order $s$. 
\end{definition}
\begin{theorem}[{\cite[Thm.\ 6.12]{Foucart13}}]
    \label{thm:RIP}
    Suppose the RIC $\delta_{2s}$ of a matrix
$A \in \C^{m\times N}$ satisfies $\delta_{2s} < 4/\sqrt{41}\approx 0.62$. Then 
for any $x\in\C^N$ and $e \in \C^m$ with $\|e\|_2\leq \eta$, any solution
$\hat{x} \in \C^N$ of 
\[ \minimize_{z \in \C^N} \|z\|_1 \quad\text{subject to}\quad
\|z - (Ax+e)\|_{2} \leq \eta \] 
satisfies 
\[
\|x-\hat{x}\|_2 \leq \frac{C}{\sqrt{s}} \sigma_{s}(x)_1 + D \eta
\]
where $C,D> 0$ are constants dependent on $\delta_{2s}$ only and 
$\sigma_{s}(x)_1 = \inf\{\|x-z\|_1 : z \in \Sigma_s \}$.
\end{theorem}

For an isometry $U\in \C^{N\times N}$the question of whether or not $P_{\Om} U$ satisfies the RIP is related to the so-called \emph{coherence} of $U$:
\begin{definition}[Coherence] Let $U\in \C^{N\times N}$ be an isometry. The 
\emph{coherence} of $U$ is 
\[ \mu(U) = \max_{i,j =1,\ldots, N} |U_{ij}|^2 \in [N^{-1},1]. \]
\end{definition}
\begin{theorem}[{\cite[Thm.\ 12.32]{Foucart13}}]
    Let $U \in \C^{N\times N}$ be an isometry and let $0< \delta,\eps < 1$. Suppose 
$\Om = \{t_1, \ldots t_m\} \subseteq \{1,\ldots, N\}$ where each $t_k$  is
chosen uniformly and independently at random from the set $\{1,\ldots,N\}$. If 
\[ m \gtrsim \delta^{-2}\cdot s \cdot N \cdot \mu(U)\cdot \(\log(2m)\log(2N)\log^2(2s) + \log(\eps^{-1})\)\]
then with probability $1-\eps$ the matrix $A = \tfrac{1}{\sqrt{p}}P_{\Om}U\in \C^{m\times N}$, with $p=\tfrac{m}{N}$, satisfies the 
RIP of order $s$ with $\delta_s \leq \delta$.
\end{theorem}

(We slightly abuse notation here in that we allow for possible repeats of the
values $t_i$ that make up $\Omega$).
Thus if the coherence $\mu(U)\approx N^{-1}$ we obtain the RIP
of order $s$ using approximately $s$ measurements up to constants and log factors.

There are, however, two problems with this approach.  First, in our setup, where $U = V \Psi^{T}$ is the product of a Fourier or Hadamard matrix and a discrete wavelet transform, the coherence $\mu(U)\approx 1$.  Hence satisfying the RIP requires at least $m \approx N$ measurements.  Second, the RIP asserts recovery for \textit{all} $s$-sparse vectors of wavelet coefficients, and thus does not exploit any additional structure these coefficients possess.  However, as stated, wavelet coefficient are highly structured: large wavelet coefficients tend to cluster at coarse scales, with coefficients at fine scales being increasingly sparse.

Motivated by this, the following structured sparsity model was introduced in \cite{Adcock17}:

\begin{definition}[Sparsity in levels]
Let $\M = [M_1,
\ldots, M_r] \in \N^{r}$, $M_0 = 0$, with $1 \leq M_1 < \cdots < M_r=M$ and let 
$\s = (s_{1},\ldots, s_r) \in \N^r$ with 
$s_{l}\leq M_{l}-M_{l-1}$, for $l=1,\ldots, r$.
We say that the vector $x \in \C^M$ is sparse in levels if 
\[ |\supp{(x)}\cap \{M_{l-1}+1,\ldots, M_l\}| \leq s_{l}\quad\text{for } l =1,\ldots,r.\]
In which case we call $x$, $(\s,\M)$-sparse, where $\s$ and $\M$ are called the local sparsities and sparsity levels, respectively. We denote the set of all $(\s,\M)$-sparse vectors 
by $\Sigma_{\s,\M}$.
\end{definition}

As noted above, randomly subsampling an isometry $U$ is a poor measurement protocol for coherent problems such as Fourier--Wavelets.  Instead, in \cite{Adcock17} it was proposed to sample in the following structured way:

\begin{definition}[Multilevel random subsampling]
Let $\Nm = [N_1, \ldots, N_r] \in \N^{r}$, where $1 \leq N_1 < \cdots < N_r =
N$ and $\m = (m_1,\ldots, m_r) \in \N^r$ with $m_k \leq N_{k}-N_{k-1}$ for
$k=1,\ldots, r$, and $N_0=0$. For each $k=1,\ldots,r$, let $\Omega_{k} = \{ 
N_{k-1}+1,\ldots, N_{k} \}$ if $m_{k} = N_{k}-N_{k-1}$ and if not, let $t_{k,1}, \ldots,
t_{k,m_k}$ be chosen uniformly and independently from the set $\{N_{k-1}+1,
\ldots, N_k\}$, and set $\Om_{k}= \{ t_{k,1}, \ldots, t_{k,m_k} \}$. If
$\Om = \Om_{\Nm,\m} = \Om_{1} \cup \cdots \cup \Om_{r}$ we refer to 
$\Om$ as an $(\Nm, \m)$\emph{-multilevel subsampling scheme}. 
\end{definition}

For this structured model, the following extensions of the RIP was
first introduced in \cite{Bastounis17}. 

\begin{definition}[RIPL]
\label{def:RIPL}
    Let $\s,\M \in \N^r$ be given local sparsities and sparsity levels, respectively. 
For a matrix $A \in \C^{m\times N}$ the \emph{Restricted Isometry Constant
in Levels (RICL)} of order $(\s,\M)$, denoted 
$\delta_{\s,\M} $, is the smallest $\delta \geq 0$ such that 
\[ (1-\delta)\|x\|_{2}^{2} \leq \|Ax\|_{2}^{2} \leq (1+\delta)\|x\|_{2}^{2} 
\quad \forall x \in \Sigma_{\s,\M}. \] 
We say that $A$ has the \emph{Restricted Isometry Property in Levels (RIPL)} if $0 \leq \delta < 1$.
\end{definition}

We shall see that this leads to uniform recovery of all $(\s,\M)$-sparse
vectors, but first we define the \emph{best $(\s,\M)$-term approximation
error} of $x \in \C^{N}$. That is  
\[\sigma_{\s,\M}(x)_p \coloneqq \inf\{ \|x-z\|_{p} : z \in \Sigma_{\s,\M} \}.\]
\begin{theorem}[{\cite[Thm.\  4.4]{Bastounis17}}]
\label{thm:RIPL_imply_uniform}
    Let $\s,\M \in \N^r$ be local sparsities and sparsity levels, respectively.
Let $\alpha_{\s,\M} = \max_{k,l=1,\ldots,r} s_{l}/s_{k}$ and $s = s_1 + \cdots
+ s_r$. Suppose that the RICL $\delta_{2\s,\M}\geq 0$ for the
matrix $A \in \C^{m\times M}$ satisfies 
\begin{equation}
    \label{eq:RIPL_delta}
    \delta_{2\s,\M} < \frac{1}{\sqrt{r(\sqrt{\alpha_{\s,\M}}+\tfrac{1}{4})^2+1}}.
\end{equation}  
Then, for  $x\in \C^{M}$ and $e \in \C^{m}$ with $\|e\|_{2}\leq \eta$, any solution
$\hat{x}$ of 
\[ \minimize_{z \in \C^M} \|z\|_1 \quad\text{subject to}\quad
\|z - (Ax+e)\|_{2} \leq \eta \] 
satisfies
\[
\|x-\hat{x}\|_2 \leq 
(C + C' (r\alpha_{\s, \M})^{1/4}) \frac{\sigma_{s,\M}(x)_1}{\sqrt{s}} 
+ (D + D' (r\alpha_{\s, \M})^{1/4}) \eta
\]
where $C,C', D, D' > 0$ are constants which only dependent on $\delta_{2\s,\M}$.
\end{theorem}

In \cite{Li17} the authors investigated conditions under which a subsampled isometry $U\in
\C^{N\times N}$ satisfies the RIPL. In was shown that the number of samples
required to satisfy the RIPL was related 
to the so-called \emph{local coherence} properties of $U$:
\begin{definition}
\label{def:local_coher_fin}
Let $U\in C^{N \times N}$ be an isometry and $\Nm,\M\in \N^{r}$ be given 
sampling and sparsity levels. The \emph{local coherence} of $U$ is 
\[ \mu_{k,l} = \mu_{k,l}(\Nm,\M) = \{\max |U_{ij}|^2: i = N_{k-1}+1,\ldots, N_{k}, 
j = M_{l-1}+1,\ldots, M_{l}\}. \]      
\end{definition}
\begin{theorem}[{\cite[thm. 3.2]{Li17}}]
    \label{thm:Li_RIPL_mk}
    Let $U \in \C^{N\times N}$ be an isometry. Let $r\in\N$, $0<\delta,\eps < 1$,
and $0\leq r_0 \leq r$. Let $\Om = \Om_{\Nm,\m}$ be an $(\Nm,\m)$-multilevel random 
subsampling scheme. Let $\tilde{m} = m_{r_0+1}+\ldots + m_r$ and $s = s_1 +
\ldots +s_r$. Suppose that the $m_k$s satisfy
\begin{equation}
\label{ParksCanada}
m_{k}= N_{k}-N_{k-1},\quad \text{for }k=1,\ldots,r_0,
\end{equation}
and 
\begin{equation}
    \label{eq:Li_RIPL_mk}
    m_{k}\gtrsim  \delta^{-2}\cdot (N_{k}-N_{k-1})\cdot \(\sum_{l=1}^{r}s_{l}\mu_{k,l}\)
\cdot \(r  \log(2\tilde{m}) \log(2N)\log^2(2s)+\log(\eps^{-1})\)
\end{equation}
for $k=r_0+1,\ldots,r$. Then the
matrix
\begin{equation}
A = \begin{bmatrix} \tfrac{1}{\sqrt{p_1}}P_{\Om_1}U\\ \vdots \\ \tfrac{1}{\sqrt{p_r}}P_{\Om_r}U\end{bmatrix}
\quad\text{where }p_{k} = \frac{m_k}{N_{k}-N_{k-1}}\quad\text{for }k=1,\ldots,r 
\label{def:A_finite}
\end{equation}
satisfies the RIPL of order $(\s,\M)$ with constant $\delta_{\s,\M} \leq \delta$.
\end{theorem}

This theorem characterizes the number of local measurements $m_k$ needed
to ensure uniform recovery explicitly in terms of local sparsities $s_k$ and
local coherences $\mu_{k,l}$.  In particular, if the local coherences are
suitably well-behaved, then recovery may still be possible from highly
subsampled measurements, even though the global coherence may be high (see
next).  Note that the condition \eqref{ParksCanada}, whereby the first $r_0$
sampling levels are saturated, models practical imaging scenarios where the low
Fourier frequencies are typically fully sampled.

To illustrate this theorem, in \cite{Adcock16} the authors consider the
one-dimensional discrete Fourier sampling problem with sparsity in Haar
wavelets.  For the Haar wavelet basis we choose an ordering where the
first level $\{ M_{0}+1, M_1 \} = \{ 1,2\}$ consists of the scaling function
and mother wavelet and the subsequent levels are chosen so that $\{M_{l-1}+1,
\ldots, M_{l}\} = \{2^{l-1}+1, \ldots, 2^{l}\}$ consists of the wavelets at
scale $l-1$. This gives the sparsity levels
\[ \M = [2^{1}, 2^{2}, \ldots, 2^{r} ], \]
where $r = \log_2(N)$ (assumed to be an integer).
Next we define the entries in the Fourier matrix $V_{\text{Four}} \in \C^{N\times N}$
as 
\[  (V_{\text{Four}})_{\om = -N/2+1, ~j=1}^{N/2,~ N}  = \frac{1}{\sqrt{N}} 
\exp (2\pi i(j-1)\om/N ), \]
where we have started the ordering of the rows with negative indices 
for convenience. We define the sampling levels for the frequencies $\om$
in dyadic bands with $W_1  = \{0,1\}$ and 
\[ W_{k+1} = \{ -2^k+1, \ldots, -2^{k-1} \} \cup \{ 2^{k-1}+1, \ldots, 2^{k} \}, 
\quad k = 1,\ldots,r-1. \]
Notice that for a suitable reordering of the rows of $V_{\text{Four}}$ these
bands corresponds to the sampling levels $\Nm = [2^1,2^2,\ldots, 2^r]$. 
\begin{theorem}[{\cite[Cor.\ 3.3]{Li17}}]
    Let $N = 2^{r}$ for some $r \geq 1$ and let $U = V_{\text{Four}}
\Psi^{-1}\in \C^{N\times N}$, where $\Psi$ is the Haar wavelet
matrix. Let $0 < \delta, \eps < 1$ and let $\Nm = \M = [2^1, \ldots, 2^r]$. Let
$m = m_1 + \cdots m_r$ and $s = s_1 + \cdots s_r$. For each $k=1,\ldots, r$
suppose we draw $m_k$ Fourier samples from band $W_k$ randomly and
independently, where 
\[ m_k \gtrsim \delta^{-2} \cdot \bigg(\sum_{l=1}^{r} 2^{-|k-l|}s_l \bigg)
\(r \log(2 m)\log(2N)\log^2(2s)  + \log(\eps^{-1})\).\] Then with probability
at least $1-\eps$ the matrix \eqref{def:A_finite} satisfies the RIPL with
constant $\delta_{s,\M} \leq \delta$. 
\end{theorem}

Here, for convenience, we have taken $r_0 = 0$; see \cite{Li17} for further
discussion on this point.

\subsection{Shortcomings}
\label{subsec:shortcomings}

These results have two primary shortcomings, which we now discuss in further detail.
The key issue is that they are limited to finite dimensions.  As noted in Section
\ref{s:introduction}, applying finite-dimensional recovery procedures to analog
problems can result in artefacts.  For simplicity, let $N = 2^p$.
We have argued that analog signals should be modelled as elements in
$L^2([0,1))$, rather than $\C^N$.  Yet, above we have tried to use discrete
tools for recovering the signal $f \in L^2([0,1))$ by replacing $\Ws
f$ and $\Fs f$ with $V_{\text{Had}}$ and $V_{\text{Four}}$, respectively. Next
we argue that this construction leads to both \emph{measurement mismatch} and the
\emph{wavelet crime}.

Let $\chi_{[a,b)}$ denote step functions on the interval $[a,b)$ and set
$\Delta_{k,p} = [k2^{-p},(k+1)2^{-p})$.  We see that replacing $\Ws f$ with
$V_{\text{Had}} \in \C^{N\times N}$ is equivalent to replacing $f$ by e.g.
$\tilde{f} = \sum_{k=0}^{N-1} c_{k}\chi_{\Delta_{k,r}}$ for some $c \in
\C^{N}$, since $\Ws \tilde{f} = V_{\text{Had}}c$.  Clearly, $\Ws \tilde{f}$
will be a poor approximation to $\Ws f$. We refer to this as measurement
mismatch.

Next let $\phi^0, \phi^1$ denote a scaling function and wavelet,
respectively,  and set $\phi_{j,k}^s = 2^{j/2}\phi^s(2^j\cdot -k)$ for $s
\in\{0,1\}$.
By construction the solution $\hat{x}$ of (2.1) will be the coefficients of a function 
$\hat{f}$ written in a basis consisting of both wavelets and scaling functions. 
Equivalently we can represent $\hat{f}$ in the basis $\{\phi_{j,k}^0 \}_{k=0}^{N-1}$
using the coefficients $c = \Psi^{-1}\hat{x} \in \C^{N}$. The wavelet crime
is whenever we let $c$, represent pointwise samples of $f$ i.e. $c_k = f(k/N)$.

What does this mean for reconstruction?  To illustrate the issue we provide a similar example to the first numerical simulation in \cite{Adcock15gscs}, showing how finite-dimensional compressed sensing fails to recover even a function that is 1-sparse (meaning it has only one non-zero coefficient) in its wavelet decomposition. Indeed, in Figure
\ref{fig:crimes} we consider the problem of recovering a function $f$ from samples of the continuous Walsh transform. In particular, we choose $f(t) = \phi_{4,4}(t)$, where $\phi$ is
the Daubechies scaling function, corresponding to the wavelet with four vanishing
moments. Figure
\ref{fig:crimes} shows the poor performance of CS using the discrete finite-dimensional setup when
applied to a continuous problem.  Conversely, the infinite-dimensional CS
approach, which we develop in the next sections, gives a much higher fidelity
reconstruction from exactly the same samples as used in the finite-dimensional case. In fact, the infinite-dimensional CS reconstruction recovers $f$ perfectly up to numerical errors occurring from solving the optimization problem. We also observe the slightly paradoxical phenomenon in the finite-dimensional case: more samples do not improve performance. This is due to the fact that the finite-dimensional CS solution with full sampling coincides with the truncated Walsh series (direct inversion) approximation. This approximation is clearly highly suboptimal, as demonstrated in Figure \ref{fig:crimes}.

We note in passing that the above crimes stem from too early a
discretization of the inverse problem. Our infinite-dimensional CS approach
replaces $V_{\text{Had}} \Psi^{-1}$ by a finite section of the an
isometry $U \in \Bs(\ell^2(\N))$ representing change of basis between the continuous Fourier or Walsh transform
and wavelet basis.

On a related note, even if one were to ignore the above issues, estimating the local coherences $\mu_{k,l}$ in the discrete setting for anything but the Haar wavelet becomes extremely complicated.  Conversely, by moving to the continuous setting, these estimates become much easier to derive.  We do this later in the paper for arbitrary Daubechies' wavelets with the Walsh transform.

\begin{figure}[t]
\vspace{-1cm}
    \begin{tabular}{m{0.46\textwidth}m{0.46\textwidth}}
    \qquad  \,\, Infinite-dimensional CS (16 samples)&  \qquad  \,\,Truncated Walsh series (32 samples)\\
     \includegraphics[width=0.5\textwidth]{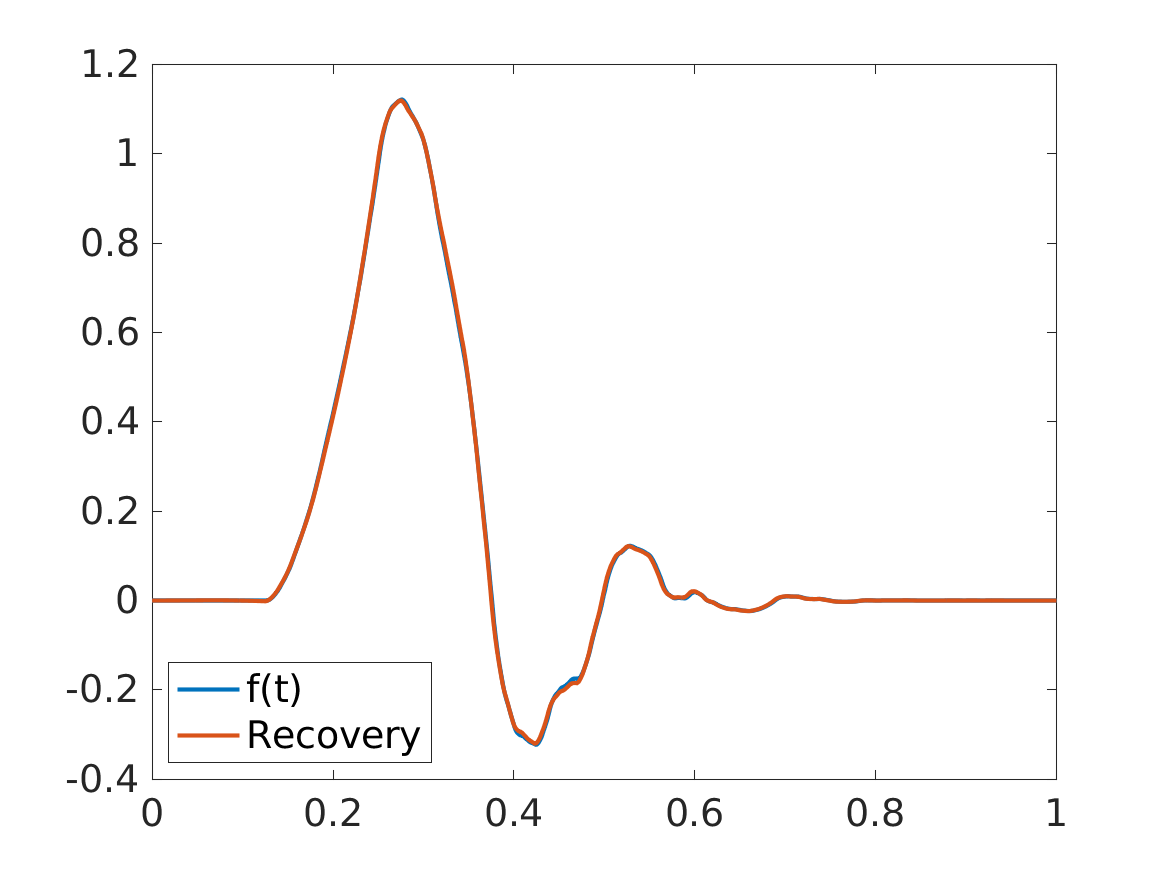} 
    &  \includegraphics[width=0.5\textwidth]{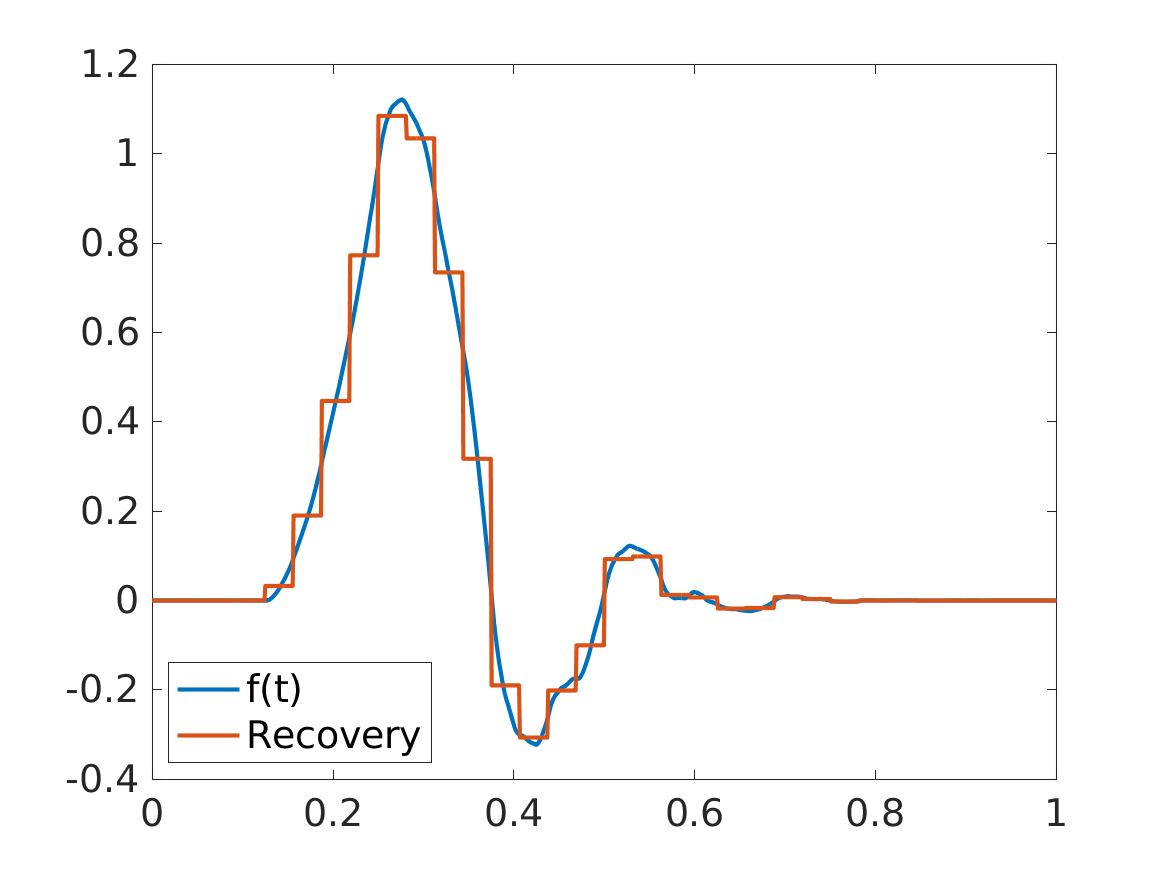}  \\
    \qquad  \,\, Finite-dimensional CS (16 samples)&  \qquad  \,\,Finite-dimensional CS (32 samples)\\
   \includegraphics[width=0.5\textwidth]{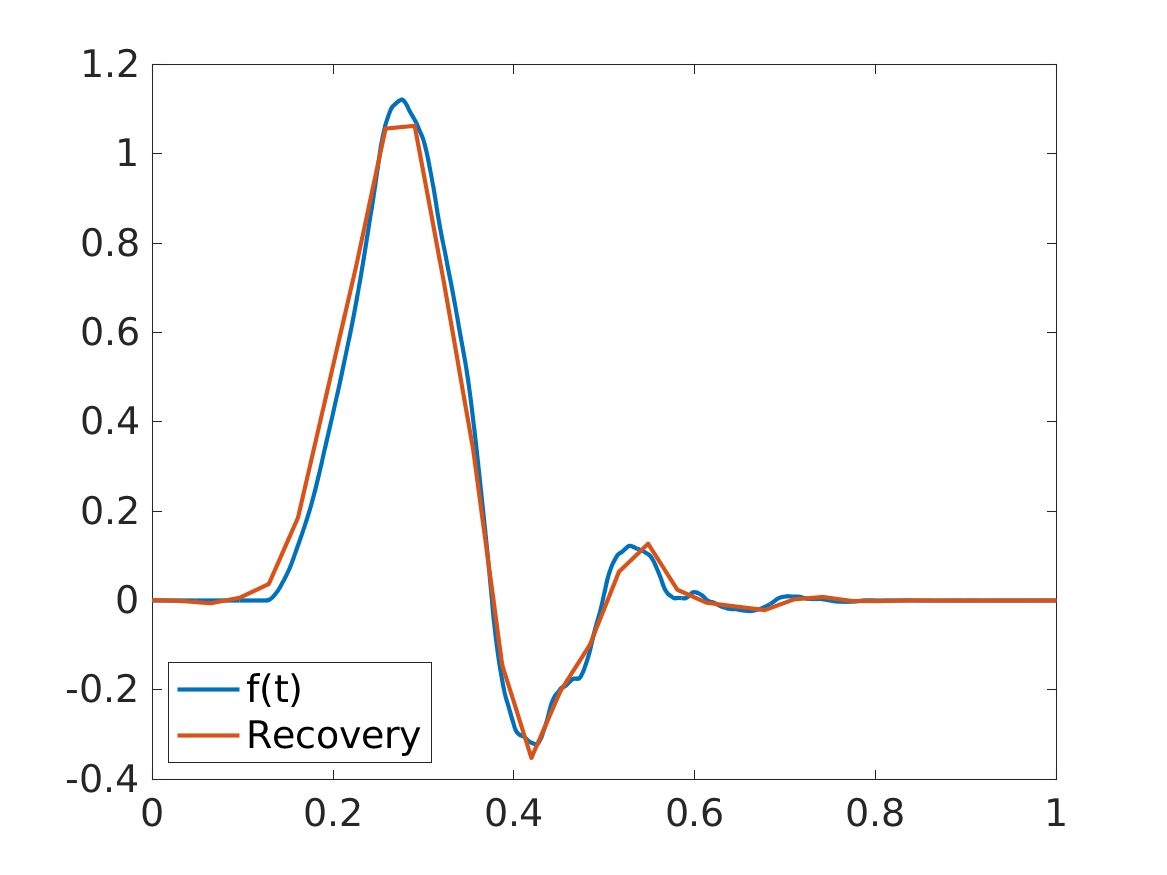} 
    & \includegraphics[width=0.5\textwidth]{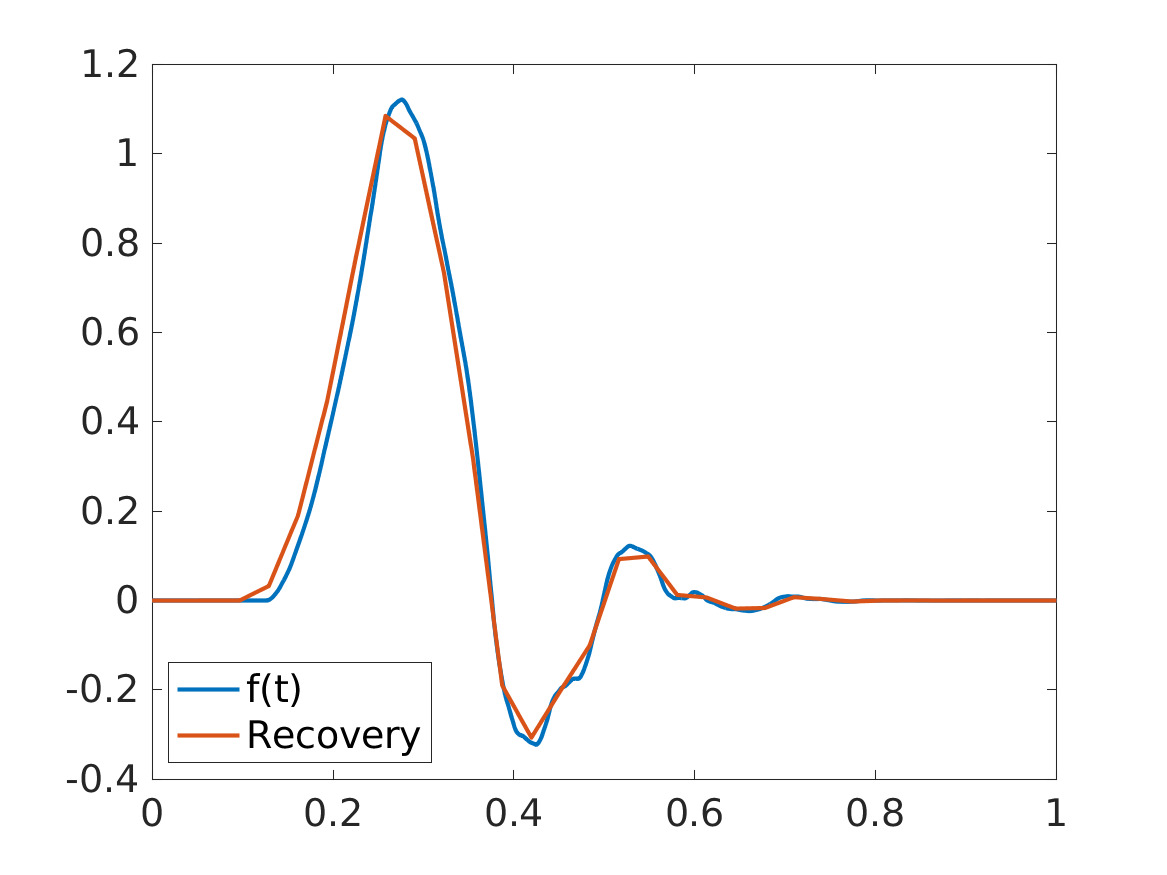} 
\end{tabular}
\caption{Reconstructions (using Walsh samples) of $f(t) = \phi_{4,4}(t)$, where $\phi$ is
the Daubechies scaling function, corresponding the wavelets with four vanishing
moments. Upper
left: Reconstruction from the first $16$ Walsh samples using an infinite-dimensional CS model. Upper right: Truncated Walsh series based on the first $32$ Walsh samples.  Lower left: Reconstruction from the first $16$ Walsh samples using the finite-dimensional ($32 \times 32$) CS model.  Lower right: Reconstruction from the first $32$ Walsh samples using the finite-dimensional ($32 \times 32$) CS model. In theory, the right images should be the same, however, numerical errors in the optimisation cause the difference.  
}
\label{fig:crimes}
\end{figure}

The second shortcoming relates to Theorem \ref{thm:RIPL_imply_uniform}. It says that we can guarantee recovery of
all sparse signals provided the matrix $A \in \C^{m \times M}$ satisfies the
RIPL with constant
\[ \delta_{2\s,\M} < \frac{1}{\sqrt{r(\sqrt{\alpha_{\s,\M}}+\tfrac{1}{4})^2+1}}. \]
Here $r$ is the number of levels and $\alpha_{\s,\M} = \max_{k,l = 1,\ldots,r}
s_{l}/s_k$ is the sparsity ratio. Inserting the above inequality into Theorem
\ref{thm:Li_RIPL_mk} gives a sampling condition of the form
\[ m_k \gtrsim r\cdot\alpha_{\s,\M} \cdot (N_{k}-N_{k-1})\cdot \(\sum_{l=1}^{r} \mu_{k,l}s_l \)\cdot L\] 
where $L$ is the log factors. This means that the sparsity ratio $\alpha_{\s,\M}$
will affect the sampling condition in all sampling levels. Thus for signals where
we expect the local sparsities to vary greatly from level to level (e.g.\ wavelets) this will lead to a unreasonably 
high number of samples. 

To overcome this problem, using an idea from \cite{Traonmilin15}, we replace
the $\ell^1$-regularizer in the optimization problem \eqref{eq:QCBP_fin} with a
weighted $\ell^1$-regularizer.  For a suitable choice of weights, this removes
the factor of $\alpha_{\s,\M}$ in the various measurement conditions.  As we
show, these guarantees are optimal up to constants and log factors.

\section{Extensions to infinite dimensions}
\label{sec:infin_dim}
\subsection{Setup}
\label{subsec:setup}
We will continue with the notation we introduced above, extended to
infinite dimensions. That is, we assume that the signal $f$ is an element of
$L^2([0,1))$. We still let  $P_{\Om}$ denote the projection onto the canonical
basis, but we now let it be an element in either
$ \Bs(\ell^{2}(\N))$ or $\Bs(\ell^{2}(\N), \C^{|\Om|})$. Similarly we call a
vector $x \in \ell^{2}(\N)$ $(\s,\M)$-sparse if $P_{M}x$ is $(\s,\M)$-sparse
and $P_{M}^{\perp} x = 0$. Here $M=M_r$ and we refer to it as the
\emph{sparsity bandwidth} of $x$. For an isometry $U \in \Bs(\ell^2(\N))$ we
define the coherence of $U$ as $\mu(U) = \sup \{ |U_{ij}|^2 : i,j\in \N \}$.

Next we describe the setup for a general sampling basis $\Bsa =
\allowbreak \{\bsa_{1}, \bsa_{2}, \bsa_{3}, \ldots,\}$ and a sparsifying basis
$\Bsp=\allowbreak \{\bsp_{1}, \bsp_{2}, \bsp_{3},\ldots,\}$, both assumed to be
orthonormal bases of $L^2([0,1))$. In Section \ref{sec:4}, we will specialize this so that $\Bsa$ is
the Walsh sampling basis and $\Bsp$ is a wavelet sparsifying basis. This will
enable us to derive concrete recovery guarantees for $f$. The setup below is,
however, completely general.  

For the two bases $\Bsa$ and $\Bsp$ we can represent $f$ using the coefficients $y
= \{\ind{f,\bsa_{n}}\}_{n\in \N}$ and $x = \{\ind{f,\bsp_{n}}\}_{n\in \N}$,
respectively. To change the representation from $\Bsa$ to $\Bsp$ we define the
following matrix. 
\begin{definition}
    Let $\Bsa = \allowbreak \{\bsa_{1}, \bsa_{2}, \bsa_{3}, \ldots,\}$ and
$\Bsp=\allowbreak \{\bsp_{1}, \bsp_{2}, \bsp_{3},\ldots,\}$ be orthonormal bases for
$L^2([0,1))$. The \emph{change of basis matrix} $U \in \Bs(\ell^2(\N))$
between $\Bsa$ and $\Bsp$ is the infinite matrix with entries
\[ U_{ij} = \ind{\bsp_{j}, \bsa_{i}} \] 
We will denote this matrix by $U = [\Bsa, \Bsp]$.
\end{definition}
Notice in particular that since $\Bsa$ and $\Bsp$ are orthonormal,
$U=[\Bsa,\Bsp]$ is an isometry on $\ell^2(\N)$ and we can write $y = Ux$.

Next let  $\Om = \Om_{\m,\Nm}$
be a given multilevel random sampling scheme with $|\Om| = m$.  We refer to $N
= N_r$ as the \emph{sampling bandwidth} of $\Omega$ (as discussed later, this
will be chosen in terms of sampling bandwidth to ensure stable truncation of
$U$).  Now define the matrix
\begin{equation}
\label{eq:H_mat}
 H \coloneqq 
\begin{bmatrix} 
    1/\sqrt{p_1} P_{\Om_1} U \\
    1/\sqrt{p_2} P_{\Om_2} U \\
    \ldots \\
    1/\sqrt{p_r} P_{\Om_r} U 
\end{bmatrix} \in \C^{m \times \infty}, \quad 
 \text{ where }\quad p_k = m_k/(N_k - N_{k-1})
\end{equation}
and we use the slightly unusual notation $\mathbb{C}^{m \times \infty}$ for the operators
$\mathcal{B}(\ell^{2}(\mathbb{N}), \mathbb{C}^{m})$.
Due to the scaling factors $1/\sqrt{p_k}$ we consider scaled noisy
measurements 
\begin{equation}
    \label{eq:y_tilde}
    \tilde{y} = D P_{\Om} y + e \in \C^m 
\end{equation}
where $D$ is a diagonal matrix with the corresponding scaling factors found in
$H$ along the diagonal and $e$ is the measurement noise.

Suppose that $x$ is approximately $(\s,\M)$-sparse with sparsity bandwidth $M$.
It is tempting to form the finite matrix $A = HP_{M} \in \C^{m \times M}$ and
solve the minimization problem
\[ 
\minimize \|z\|_1 
\quad \textnormal{ subject to } \quad 
\|Az - \tilde{y}\|_{2} \leq \eta.  
\]
However, note that the truncation of $H$ to $A$ introduces an additional
truncation error $H P_{M}^{\perp}x$.  Indeed,
\bes{
Ax- \tilde{y} = HP_{M}^{\perp}x+e,
}
and this poses a problem since for the above decoder we require $\eta \geq \|HP_{M}^{\perp}x+e\|_2$ in order for $P_M x$ to be a feasible point.  For some applications we might have a rough 
estimate of $\|e\|_2$, but any estimate of $\|HP_{M}^{\perp}x\|_2$ would
require a priori knowledge of $x$, the signal we are trying to recover.  This
is generally impossible. (We note in passing that there is some recent work
\cite{Brugiapaglia17} which derives CS recovery guarantees in the absence of feasibility of
the target vector $P_M x$, but the application of this work to the sparse in
levels model is not clear).

To overcome this issue, we will introduce a \emph{data fidelity parameter} $K
\geq M$ and assume we know $\|e\|_{2}$ so that we can let $\eta > \|e\|_2$.
Then there will always exits a $K' \geq M$ such that $P_{K}x$ lies in the
feasible set $\{ z \in \C^{K} :\| Az -\tilde{y}\|_2 \leq \eta\}$ corresponding to the augmented matrix 
\begin{equation}
\label{A_multilevel}
A = H P_{K}
\end{equation}
for all $K \geq K'$.  In practice (for the general case) it will also be impossible determine a
sufficient value for $K$, but for fixed $\eta > \|e\|_{2}$ there will always
exist such a $K$. It should, however, be noted that there are special cases, 
such as Walsh sampling and wavelet recovery, where sufficient values for $K$ are known; see Remark \ref{r:K_issue_fudge}.

This aside, as previously mentioned, we also now modify the optimization
problem to include weights.  Specifically, let $\M,\s \in \N^r$ be given
sparsity levels and local sparsities respectively.  For positive weights $\bom
= (\om_1,\ldots, \om_{r+1})$ we define
\[
\|x\|_{1,\bom} \coloneqq \sum_{l=1}^{r+1} \om_{l} \|P_{M_l}^{M_{l-1}}x\|_1 ,
\]
with $M_{r+1}=K$ for $x \in \C^{K}$.  Notice that this weighted regularizer assigns constant weights on each sparsity level.  With this in hand, our recovery procedure is 
\[ 
\minimize \|z\|_{1,\bom} 
\quad \textnormal{ subject to } \quad 
\|Az - \tilde{y}\|_{2} \leq \eta,
\]
with $A$ as in \eqref{A_multilevel} and $\eta \geq \| A x - \tilde{y} \|_2$.

\subsection{The balancing property}

We now discuss the relation between the sampling and sparsity bandwidths $N$ and $M$.  From 
generalized sampling theory \cite{Adcock15gscs} 
we know that we must choose $N \geq M$ to obtain a
stable mapping between the first $N$ sampling basis functions and the first $M$
sparsity basis functions. The degree of stability for this solution will depend
of the so-called \textit{balancing property}:
\begin{definition}
    Let $U \colon \ell^2(\N) \to \ell^2(\N)$ be an isometry. Let $0 < \theta <
1$ and $N \geq M \geq 1$. Then $U$ has the \emph{balancing property} with
constant $\theta$ if \[ \|P_{M}U^*P_N U P_M - P_M\|_{2} \leq 1 -\theta.  \]
\end{definition}

Note that the balancing property may not hold for any $N \geq M$.
However, it always holds for sufficiently large $N$ (for fixed $M$).  Indeed,
$P_M U^* P_N U P_M \rightarrow P_M U^* U P_M \equiv P_M$ in the operator norm,
hence the balancing property holds with $\theta$ arbitrarily close to $1$ for
large enough $N$.

Below we shall see that this property will also affect our recovery guarantees,
but it will be camouflaged as the quantity $\|G^{-1}\|_{2}$, where $G =
\sqrt{P_MU^*P_NU_PM}$. This gives the following relation.
\begin{lemma}
\label{lem:bal_prop}
    Let $U \in \Bs(\ell^2(\N))$ be an isometry satisfying the balancing
property of order $0 < \theta < 1$ for $M, N \in \N$. Let $G =
\sqrt{P_MU^*P_NUP_M}$ be self-adoint and nonnegative definite. Then
$G$ is invertible and
\begin{equation}
 \|G^{-1}\|_{2} \leq 1/\sqrt{\theta}  
\end{equation}
\end{lemma}

\subsection{$\bs{G}$-adjusted Restricted Isometry Property in Levels (G-RIPL)}
\label{subsec:GRIPL}

Our theoretical analysis requires a RIP-type property for the matrix $H P_M$.  However, as implied in the previous discussion, the finite matrix $P_N U P_M \in \C^{N \times M}$ (from which $A P_M$ is constructed) is not an isometry for any $N \geq M$.  In particular, unlike in finite dimensions $\mathbb{E}(P_M H^* H P_M ) = P_M U^* P_N U P_M = G^2$ is not the identity.  In order to handle this situation, we introduce the following generalization of the RIP:

\begin{definition}[G-RIPL]
\label{def:GRIP}
Let $A \in \mathbb{C}^{m \times M}$,  $G \in \mathbb{C}^{M \times M}$ be invertible, $\mb{M}
= (M_1,\ldots,M_r)$ be sparsity levels and $\mb{s} = (s_1,\ldots,s_r)$ be local
sparsities.  The $\mb{s}^{\rth}$ \textit{$G$-adjusted Restricted Isometry Constant in Levels (G-RICL)} $\delta_{\mb{s},\mb{M}}$ is the smallest $\delta \geq 0$
such that
\begin{equation*}
(1-\delta) \| G x \|^{2}_{2} \leq \| A x \|^{2}_{2} \leq (1+\delta) \| G x \|^{2}_{2},\quad
\forall x \in \Sigma_{\mb{s},\mb{M}}.
\end{equation*}
If $0 < \delta_{\mb{s},\mb{M}} < 1$ we say that the matrix $A$ satisfies the
\emph{$G$-adjusted Restricted Isometry Property in Levels (G-RIPL)} of order
$(\mb{s},\mb{M})$. 
\end{definition}
The G-RIPL is of course completely general and can be stated for any
$G$. However, in the following we will let $G = \sqrt{P_M U^* P_N U P_M}$ and
show that the matrix $A = HP_{K}$ (or equivalently, $H P_M$ -- note that
$\Sigma_{\s,\M}$ consists of vectors $z$ with $P^{\perp}_M z = 0$) satisfies
the G-RIPL for this particular $G$. 

First, however, we show that 
the G-RIPL implies uniform recovery.  For this, we introduce the following notation:
\[\Sboms \coloneqq \sum_{l=1}^{r} \om^{2}_{l}s_l 
\quad \text{ and }\quad \zetas = \min_{l\in\{1,\ldots,r\}} \om_{l}^{2}s_l.
\]
Notice in particular that for the choice $\bom = (1,\ldots, 1,\om_{r+1})$ we
have $\Sboms = s_1 +\ldots +s_r$ and for the choice $\bom = (s^{-1/2}_{1},
\ldots, s^{-1/2}_{r}, \om_{r+1})$ we have $\Sboms = r$.
Finally, we let $\kappa(G) = \|G\|_2 \|G^{-1}\|_2$ denote the condition number of $G$.

\begin{theorem}
\label{thm:GRIPL_imply_uniform}
Let $A \in \C^{m\times K}$, $G \in \C^{M\times M}$ with $K \geq M$ and let $\M,\s \in \N^r$
be given sparsity levels and local sparsities, respectively. Let $\bom \in \R^{r+1}$ 
be positive weights. Suppose $AP_{M}$ satisfies
the G-RIPL of order $(\t,\M)$ with constant $\delta_{\t,\M} \leq 1/2$ and

\begin{equation}
    t_{l} = \min \left\{M_{l}-M_{l-1},\, 2 \ceil{ \frac{4\kappa(G)^2
    \Sboms}{\om^{2}_{l}}}\right\} \quad \text{for } l = 1,\ldots, r.
\label{eq:min_tl}
\end{equation}
Let
\[
\om_{r+1}\geq
\sqrt{\Sboms}(\tfrac{1}{3}(1+(\Sboms/\zetas)^{1/4})^{-1} +2\sqrt{2}\|AP_{K}^{M}\|_{1\to2} \|G^{-1}\|_2).
\]
Let $\eta \geq 0$, $x \in
\C^{K}$, $e\in \C^{m}$ with $\|e\|_{2} \leq \eta$ and set $y = Ax + e$. Then any solution $\hat{x}$ of the
optimization problem
\begin{equation}
 \minimize_{z \in \C^{K}} \|z\|_{1,\bom} \quad\text{subject to} \quad 
   \|Az - y\|_2 \leq \eta 
   \label{eq:QCBP_542} 
\end{equation}
satisfies
\begin{align}
\label{eq:error_bound1}
\|x-\hat{x}\|_{1,\bom} &\leq C \sigma_{\s, \M}(x)_{1,\bom} 
                            + D \|G^{-1}\|_{2}\sqrt{\Sboms} \eta \\
\|x-\hat{x}\|_{2} &\leq (1+(\Sboms/\zetas)^{1/4}) \( C \frac{\sigma_{\s,\M}(x)_{1,\bom}}{\sqrt{\Sboms}} + D\|G^{-1}\|_{2}\eta \)
\label{eq:error_bound2}
\end{align}
where $C = 2(2+\sqrt{3})/(2-\sqrt{3})$, $D = 8\sqrt{2}/(2-\sqrt{3})$ and
$\sigma_{\s,\M}(x)_{1,\bom} = \inf  \{ \| x - z \|_{1,\bom} : z \in \Sigma_{\s,\M} \}$.
\end{theorem}

Notice that the condition on $\delta$ in the above theorem is
fundamentally different from the condition found in Theorem
\ref{thm:RIPL_imply_uniform}. In the latter one requires 
$\delta_{2\s,\M} < (r(\sqrt{\alpha_{\s,\M}}+\tfrac{1}{4})^2+1)^{-1/2}$ where
$\alpha_{\s,\M} = \max_{k,l=1,\ldots,r} s_k/s_l$ is the sparsity ratio. Thus
for sparsity levels where the local sparsities vary greatly, this bound will be
unreasonably small.

In the above theorem we have removed this sparsity ratio term, by setting
$\delta = 1/2$, and require  $\delta_{\t,\M} \leq \delta$ where $t_l \geq 2
\ceil{4\kappa(G) \Sboms w_{l}^{-2}}$. For the unweighted case this leads to a
condition of the form  
\[ t_l \geq 2\ceil{4\kappa(G)^2 (s_1 + \ldots + s_r)} \]
which could be difficult to fulfill in practice, since each $t_l$ would have to be greater 
than the total sparsity of the signal. However, by considering 
the weights $\bom = (s_{1}^{-1/2},\ldots, s_{r}^{-1/2},\om_{r+1})$ we obtain a
condition of the form 
\[ t_l \geq 2\ceil{4\kappa(G)^2 r s_l}, \]
 where $t_l$ is independent of $s_k$ for $k \neq l$. This
means that we can write the requirement as $\delta_{2\ceil{4\kappa(G)^2 r \s}, \M} \leq
1/2$, and ignore any dependence between the $\s$-values, as was the problem in
Theorem \ref{thm:RIPL_imply_uniform}. 

\subsection{Sufficient condition for the G-RIPL}

In Definition \ref{def:local_coher_fin} we defined the local coherence
$\mu_{k,l}$ of an isometry $U \in \C^{N\times N}$. We extend this to isometries
$U \in \Bs(\ell^2(\N))$ in the exact same way
\[ \mu_{k,l} = \mu_{k,l}(\Nm,\M) = \{\max |U_{ij}|^2: i = N_{k-1}+1,\ldots, N_{k}, 
j = M_{l-1}+1,\ldots, M_{l}\}. \]      
This yields the following theorem.

\begin{theorem}
[Subsampled isometries and the G-RIPL]
\label{t:RIPlevels1}
Let $U \in \mathcal{B}(\ell^2(\mathbb{N}))$ be an isometry, and let $\Om =
\Omega_{\mb{N},\mb{m}}$ be an $(\mb{N},\mb{m})$-multilevel sampling scheme with
$r$ levels. Let $\M,\s \in \N^r$ be sparsity levels and local
sparsities, respectively. Let $\eps,\delta \in (0,1)$ and let $0 \leq r_0 \leq
r$, with $\tilde{m}=m_{r_0+1} + \cdots + m_r$. Let $s = s_1 + \cdots + s_r$ and 
$L = \Llog$. 
Suppose $G = \sqrt{P_{M}U^*P_N UP_M}$ is non-singular. If
\begin{equation}
\label{eq:RIPL1_cond1}
m_k = N_k - N_{k-1},\quad k=1,\ldots,r_0,
\end{equation}
and
\begin{equation}
\label{eq:GRIPL_sampdens}
m_k \gtrsim \delta^{-2}  \cdot \|G^{-1}\|^{2}_{2}\cdot (N_k - N_{k-1}) \cdot \bigg (
\sum^{r}_{l=1} \mu_{k,l} \cdot s_l \bigg ) \cdot L,
\end{equation}
for $k=r_0+1,\ldots,r$ then with probability at least $1-\epsilon$, the matrix
\begin{equation}
A = \begin{bmatrix} 1/\sqrt{p_1}P_{\Om_1}UP_{M}\\ \vdots \\ 1/\sqrt{p_r}P_{\Om_r}UP_{M}\end{bmatrix}
\quad\text{where }p_{k} = \frac{m_k}{N_{k}-N_{k-1}}\quad\text{for }k=1,\ldots,r 
\label{def:A_infinite}
\end{equation}
satisfies the G-RIPL of order $(\mb{s},\mb{M})$
with constant $\delta_{\mb{s},\mb{M}} \leq \delta$.
\end{theorem}

\subsection{Overall recovery guarantee}

Theorem \ref{thm:GRIPL_imply_uniform} and Theorem \ref{t:RIPlevels1} yield the next results.

\begin{corollary}
\label{cor:overall}
Let $U \in \mathcal{B}(\ell^2(\mathbb{N}))$ be an isometry, and let $\Om =
\Omega_{\mb{N},\mb{m}}$ be an $(\mb{N},\mb{m})$-multilevel sampling scheme with
$r$ levels. Let $\M,\s \in \N^r$ be sparsity levels and local sparsities,
respectively, and let $\bom = [s_{1}^{-1/2}, \ldots, s_{r}^{-1/2}, \om_{r+1}]$
be weights. Let $\eps,\delta \in (0,1)$ and $0 \leq r_0 \leq r$. Let
$m=m_1+\ldots+m_r$, $\tilde{m}=m_{r_0+1} + \cdots + m_r$, $s = s_1 + \cdots +
s_r$, and $L = \Llog$. Let $H \in \C^{m\times \infty}$ be as in \eqref{eq:H_mat}
and set $A = HP_{K}$. Let $x \in \ell^{2}(\N)$, $e_1 \in \C^{m}$ and $\eta > 0$.
Set $e = HP_{K}^{\perp} x + e_1$ and $\tilde{y} = A x + e$. 
Suppose 
\begin{enumerate}[(i)]
    \item we choose $M$ and $N$ so that $U$ satisfies the balancing property of order $0<\theta<1$,
    \item we choose $\eta \geq \|e_1\|$ and $K$ so that $\|HP_{K}^{\perp}x\|_{2} \leq \eta'$,  
    \item the weight $\om_{r+1}$ satisfies 
    \[ \om_{r+1}\geq \sqrt{r}\(\frac{1}{3(1+r^{1/4})} 
                   +2\sqrt{\frac{2}{\theta}}\|AP_{K}^{M}\|_{1 \to 2} \), 
 \]
\item the $m_k$'s satisfy $m_{k} = N_{k} - N_{k-1}$ for $k=1, \ldots, r_0$ and 
\begin{equation}
    m_{k} \gtrsim \theta^{-2} \cdot r \cdot (N_{k}-N_{k-1}) \cdot \bigg ( \sum^{r}_{l=1} \mu_{k,l} s_l \bigg) \cdot L \quad \text{ for } k = r_0+1,\ldots,r. 
\end{equation}
\end{enumerate}
Then with probability $1-\eps$ any solution $\hat{x}$ of the optimization problem
\[ \minimize_{z \in \C^{K}} \|z\|_{1,\bom} \quad\text{subject to} \quad 
   \|Az - \tilde{y}\|_2 \leq \eta + \eta'\]
satisfies
\begin{align} \label{eq:err_bd1}
\|P_{K}x-\hat{x}\|_{1,\bom} &\leq C \sigma_{\s, \M}(P_K x)_{1,\bom} 
                            + D \frac{\sqrt{r}}{\sqrt{\theta}}(\eta+\eta') \\
\|P_{K}x-\hat{x}\|_{2} &\leq (1+r^{1/4}) \( C \frac{\sigma_{\s,\M}(P_{K}x)_{1,\bom}}{\sqrt{r}} + D\frac{1}{\sqrt{\theta}}(\eta+\eta') \)
\label{eq:err_bd2}
\end{align}
where $C = 2(2+\sqrt{3})/(2-\sqrt{3})$ and 
$D = 8\sqrt{2}/(2-\sqrt{3})$.
\end{corollary}

Suppose that $x$ is exactly $(\s,\M)$-sparse.  Then the above theorem
guarantees exact recovery of $x$ via weighted $\ell^1$ minimization subject to
the corresponding measurement condition.  We note in passing this measurement
condition is optimal up to log factors, in the sense that it is the same of
that of the oracle estimator based on \textit{a priori} knowledge of
$\supp(x)$.  See \cite{ABB18}.

\section{Recovery guarantees for Walsh sampling with wavelet reconstruction}
\label{sec:4}

Having presented the abstract infinite-dimensional CS framework in full generality, the remainder of the paper is devoted to its application to the case of binary sampling with the Walsh transform with sparsity in orthogonal wavelet bases.  We first describe the setup, before presenting the main recovery guarantees in Sections \ref{subsec:rec_1d} and \ref{subsec:rec_Haar}.

\subsection{Walsh functions}
\label{sec:walshFunc}

For any number $n \in \Z_{+} = \{0,1,2, \ldots\}$ there exits a unique
dyadic expansion \[n = n_1 2^{0} + n_2 2^{1} + \ldots + n_{j}2^{j-1} + \cdots
\] where $n_j \in \{0, 1\}$ for $j \in \N$. Similarly any $x \in [0,1)$ can
be written in its dyadic form as 
\[ x = x_1 2^{-1} + x_2 2^{-2} + \cdots + x_j  2^{-j} \] 
with $x_j \in \{0,1\}$ for all $j \in \N$.  For a dyadic rational number $x$ this
expansion is not unique, as one may use either a finite expansion, or an
infinite expansion where $x_i = 1$ for all $i \geq k$ for some $k \in \N$. In
such cases we always consider the finite expansion. 
In practice this means that
we have removed countably many singletons from $[0,1)$.

\begin{definition}
    \label{def:walsh_func}
    Let $n \in \Z_{+}$ and $x \in [0,1)$. The \emph{Walsh function} 
    $w_n\colon [0,1) \to \{+1, -1\}$ is given by
    \begin{equation}
         w_n(x) \coloneqq (-1)^{\sum_{j=1}^{\infty} (n_j + n_{j+1})x_j}
    \end{equation}
\end{definition}
On the interval $[0,1)$ the Walsh function $w_n$ has $n$ sign changes, 
$n$ is therefore often denoted the \emph{frequency} of $w_n$.
The $2^r$ first Walsh functions gives rise to the entries in the  
sequency ordered Hadamard matrix 
\[(V_{\text{Had}})_{i,j} = w_{i-1}((j-1)/2^{r})\]
where $i,j = 1, \ldots, 2^r$.

\begin{definition}[Walsh basis]
    \label{def:walsh_basis}
    Define the \emph{Walsh basis} as 
    \[ B_{\textnormal{wh}} \coloneqq \{w_n : n \in \Z_+\} \]
    where \textquote{wh} is an abbreviation for Walsh-Hadamard.
\end{definition}

Note that this is an orthonormal basis of $L^2([0,1))$.

%
%
%


\subsection{Wavelet transform}
\label{subsec:wavelets1}
Let $\phi \colon \R \to \R$ and $\psi \colon \R \to \R$ be a orthonormal scaling
function and wavelet \cite{Daubechies92}, respectively, with minimal support,
corresponding to an multiresolution analysis (MRA). Note that this could 
both be the classical \textquote{Daubechies wavelet} with a minimum-phase or 
\textquote{symlets} which are close to being symmetric, but with a larger 
phase \cite[294]{Mallat09}.
Let 
\begin{equation}
\phi_{j,k} (x) \coloneqq 2^{j/2} \phi(2^j x -k ) ~~~ \text{ and }~~~ 
\psi_{j,k} (x) \coloneqq 2^{j/2} \psi (2^j x - k)
\label{eq:db_scale} 
\end{equation}
denote the scaled and translated versions. 

A wavelet $\psi$ is said to have $\nu$ vanishing moments if
\[ \int_{-\infty}^{\infty} x^k \psi(x) \d x = 0 ~~~~~ \text{ for } 0 \leq k < \nu. \]

For for orthogonal wavelets with minimum support, the support depends on the number of vanishing moments.
That is
\begin{equation}
    \Supp (\phi) = \Supp (\psi) = [-\nu+1, \nu].
\end{equation}
While this system constitutes an orthonormal basis of $L^2(\R)$, in our case we require an orthonormal basis of $L^2([0,1))$. There
exists several construction of wavelets on the interval, but we will only
consider periodic extensions and the orthogonal \emph{boundary wavelets} introduced by
Cohen, Daubechies and Vial in \cite{Cohen93}, which preserves the number of
vanishing moments. 
 
For wavelets on the interval we need to 
replace the $2\nu$ wavelets/scaling functions intersecting the boundaries at each
scale, with their corresponding boundary-corrected counterparts. We postpone the formal definition 
of periodic and boundary wavelets until we need it, in the proof sections. But to
simplify the notation let 
\begin{align*}
    \phi_{j,k}^{0} &\coloneqq
    \begin{cases} 
        \phi_{j,k}^{\text{boundary}} & \text{for } k \in \{0,\ldots,\nu-1\} \\
        \phi_{j,k} & \text{for } k \in \{\nu, \ldots, 2^j-\nu-1\}  \\
        \phi_{j,k}^{\text{boundary}} & \text{for } k \in \{2^j -\nu,\ldots, 2^j-1\}  \\
    \end{cases}, \\ 
    \phi_{j,k}^{1} &\coloneqq
    \begin{cases} 
        \psi_{j,k}^{\text{boundary}} & \text{for } k \in \{0,\ldots,\nu-1\} \\
        \psi_{j,k} & \text{for } k \in \{\nu, \ldots, 2^j-\nu-1\}  \\
        \psi_{j,k}^{\text{boundary}} & \text{for } k \in \{2^j -\nu,\ldots, 2^j-1\}  \\
    \end{cases}, 
\end{align*}
where $\phi^{\text{boundary}}_{j,k}$ and $\psi^{\text{boundary}}_{j,k}$ are either 
a periodic wavelet/scaling function or the boundary wavelet/scaling functions
introduced in \cite{Cohen93}. For the former extension we say that
$\phi_{j,k}^{s}$, $s \in \{0,1\}$ \textquote{originate from a \emph{periodic
wavelet}} while for the latter we say that it  \textquote{originate from a
\emph{ boundary wavelet}}.

We will throughout assume $J_0 \in \Z_+$ satisfies $2^{J_0} \geq 2\nu$ for
$\nu\geq 2$ and $J_0 \geq 0$ for $\nu=1$. This will ensure that there exits at
least one $k \in \{0, \ldots, 2^{j}-1\}$ such that $\Supp(\phi_{j,k}) =
\Supp(\psi_{j,k}) \subseteq [0,1)$ for all $j \geq J_0$.
\begin{definition}
    \label{def:Bwave}
    For a fixed number of vanishing moments $\nu$, minimum wavelet decomposition
    $J_0$ and a boundary extension which is either periodic or boundary wavelets, 
    let $\phi_{j,k}^{s}$ be the corresponding wavelets and scaling functions.  
    We define 
    \begin{equation*}
        B_{\textnormal{wave}}^{J_0, \nu}
= \left \{
\phi_{J_0,0}^{0}, \ldots, \phi_{J_0,2^{J_0}-1}^{0}, 
\phi_{J_0,0}^{1}, \ldots, \phi_{J_0,2^{J_0}-1}^{1}, 
\phi_{J_0+1,0}^{1}, \ldots, \phi_{J_0+1,2^{J_0+1}-1}^{1}, 
\ldots
\right\} 
    \end{equation*}
\end{definition}
Both $\Bwh$ and $\Bwave^{J_0,\nu}$ are orthonormal bases for $L^2([0,1))$.


\subsection{Recovery guarantees}
\label{subsec:rec_1d}
From Section \ref{sec:infin_dim} there are four unknown factors depending on $U$ which need to be estimated. These are the local coherences $\mu_{k,l}$, the norm 
$\|HP_{K}^{M}\|_{1\to 2}$ where $H$ is given by \eqref{eq:H_mat}, the
condition number $\kappa(G) = \|G\|_{2}\|G^{-1}\|_{2}$ and the factor
$\|G^{-1}\|_{2}$ found in condition \eqref{eq:GRIPL_sampdens}.

For the two latter factors we have $G =\sqrt{P_{M}U^* P_N UP_M}$. Furthermore
we know that $\|G\|_{2} \leq 1$ since $U$ is an isometry. 
In practice we therefore only need to determine an upper bound
$\|G^{-1}\|_{2}$ and from Lemma \ref{lem:bal_prop} we know that $\|G^{-1}\|_{2}
\leq 1/\sqrt{\theta}$, where $0< \theta < 1$ is the balancing property
constant. In other words, it suffices to determine when the balancing
property holds with a given $\theta$.

The following three propositions estimate these quantities for the case
$U = [\Bwh, \Bwave^{J_0,\nu}]$.  
\begin{proposition}\label{lem:Terhaar_1d}
Let $U = [\Bwh, \Bwave^{J_0,\nu}]$. For each $\theta \in (0,1)$, there exits a constant 
$q_{\theta} \geq 0$, such that whenever $N = 2^{k+q_{\theta}} \geq 2^{k}= M$ then 
$U$ satisfies the balancing property of order $\theta$ for all $k \in \N$. 
\end{proposition}
Note that Proposition \ref{lem:Terhaar_1d} is a consequence of Theorem 1.1 in
\cite{Terhaar17}.

\begin{proposition}
\label{lem:1d_coher_bound}
Let $U = [\Bwh, \Bwave^{J_0, \nu}]$ with $\nu\geq 3$ and let 
\begin{equation*}
\label{eq:samp_lev_1d}
    \M = [2^{J_0+1}, \ldots, 2^{J_0 + r}] \text{ and } \Nm = [2^{J_0+1},\ldots,
2^{J_0+r-1}, 2^{J_0+r+q}]  \text{ with } q\geq 0,
\end{equation*}
be sparsity and sampling levels, respectively.
Then the local coherences of $U$ scales like 
\[  \mu_{k,l} \lesssim 2^{-J_0-k} 2^{-|l-k|}.  \]
\end{proposition}
\begin{proposition}
\label{prop:norm12}
Let $U = [\Bwh, \Bwave^{J_0, \nu}]$ and let $\M,\Nm \in \N^r$ be sparsity and sampling
levels. Let $\Om = \Om_{\m,\N}$ be a multilevel random sampling scheme, and let $H$ 
be as in \eqref{eq:H_mat}. Then 
\[\|HP_{K}^{\perp}\|_{1\to 2} \lesssim \sqrt{\frac{N}{K}}.\]  
\end{proposition}

We can now present the two main theorems in this section. We point out that 
these are only valid for $\nu \geq 3$ vanishing moments. For $\nu=1$, the 
corresponding wavelet is the Haar wavelet, and will be considered in the next
subsection.  For $\nu=2$, the coherence of $U = [\Bwh, \Bwave^{J_0, 2}]$ does
not decay as fast as for the other wavelets. Whether this is because our
coherence bounds are not sharp enough for this wavelet or if it is because 
the coherence of $U = [\Bwh, \Bwave^{J_0,2}]$ actually decays more slowly is not known. 
We do, however, present some numerics in Section \ref{subsec:sharpness_coher} 
which indicate that it is potentially the latter. 

\begin{theorem}
\label{thm:1d_GRIPL}
Let $U = [\Bwh, \Bwave^{J_0, \nu}]$ with $\nu\geq 3$ and let 
\[
\M = [2^{J_0+1}, \ldots, 2^{J_0+r}] \text{ and }~ \Nm = [2^{J_0+1},\ldots,
2^{J_0+r-1}, 2^{J_0+r+q}]  \text{ with } q\geq 0, 
\]
be sparsity and sampling levels, respectively. Let $\s \in \N^r$ be local
sparsities. Suppose $q$ is chosen so that $U$ satisfies the balancing property 
with constant $0 < \theta< 1$ and set $G = \sqrt{P_{M}U^*P_N UP_M}$. 
Let $\eps,\delta \in (0,1)$ and let $0 \leq r_0 \leq
r$, with $\tilde{m}=m_{r_0+1} + \cdots + m_r$. Let $s = s_1 + \cdots + s_r$ and 
$L = \Llog$. 
If  
\begin{equation}
\label{eq:RIPL1_cond2}
m_k = N_k - N_{k-1},\quad k=1,\ldots,r_0,
\end{equation}
and
\begin{equation*}
m_k \gtrsim \delta^{-2} \cdot \theta^{-1} \cdot 2^{q \max\{k+1-r,0\}}\cdot \bigg (
\sum^{r}_{l=1} 2^{-|k-l|} s_l \bigg ) \cdot L
\end{equation*}
for $k=r_0+1,\ldots,r$, 
then with probability at least $1-\epsilon$, the matrix
in \eqref{def:A_infinite}
satisfies the G-RIPL of order $(\mb{s},\mb{M})$
with constant $\delta_{\mb{s},\mb{M}} \leq \delta$.
\end{theorem}

With this in hand, we now present our main result:

\begin{theorem}
    \label{thm:1d_GRIPL_rec}
Let $U = [\Bwh, \Bwave^{J_0, \nu}]$ with $\nu\geq 3$
and let 
\[
\M = [2^{J_0+1}, \ldots, 2^{J_0+r}] \text{ and }~ \Nm = [2^{J_0+1},\ldots,
2^{J_0+r-1}, 2^{J_0+r+q}] , \text{ with } q\geq 0 
\]
be sparsity and sampling levels, respectively. Let $\s \in \N^r$ be local
sparsities, $\bom = (s^{-1/2}_{1},\ldots, s_{r}^{-1/2}, \om_{r+1})$ be weights
and let $\m \in \N^r$ be sampling densities.  Let $\eps \in (0,1)$ and let $0
\leq r_0 \leq r$. Let $m = m_1+\ldots + m_{r}$, $\tilde{m}= m_{r_0+1}+ \cdots+m_r$, $s = s_1 + \ldots
+ s_r$, and $L = \Llog$. 

Let $H \in \C^{m\times \infty}$ be as in
\eqref{eq:H_mat} and set $A = HP_{K}$. Let $x \in \ell^2(\N)$, $e_1
\in \C^m$ and $\eta > 0$. Set $e = HP_{K}^{\perp}x + e_1$ and   $\tilde{y} =
Ax + e$.
Suppose 
\begin{enumerate}[(i)]
    \item we choose $q=q_{\theta}$ as in Proposition \ref{lem:Terhaar_1d} so that $U$ satisfies the balancing property of order $0<\theta<1$,
    \item we choose $\eta \geq \|e_1\|$ and $K$ so that
    $\|HP_{K}^{\perp}x\|_{2} \leq \eta'$,  
    \item the weight $\om_{r+1}$ satisfies 
    \[ \om_{r+1}\geq \sqrt{r}\(\frac{1}{3(1+r^{1/4})} 
                   +2\sqrt{\frac{2}{\theta}}\|AP_{K}^{M}\|_{1 \to 2} \), 
 \]
\item the $m_k$'s satisfy $m_{k} = N_{k} - N_{k-1}$ for $k=1, \ldots, r_0$ and 
\begin{equation}
    m_{k} \gtrsim \theta^{-2} \cdot r \cdot 2^{q \max\{ k+1-r,0 \}}\bigg ( \sum^{r}_{l=1} 2^{-|k-l|} s_l \bigg) \cdot L \quad \text{ for } k = r_0+1,\ldots,r. 
\end{equation}
\end{enumerate}
Then with probability $1-\eps$ any solution $\hat{x}$ of the optimization problem
\[ \minimize_{z \in \C^{K}} \|z\|_{1,\bom} \quad\text{subject to} \quad 
   \|Az - \tilde{y}\|_2 \leq \eta + \eta'\]
satisfies
\begin{align} \label{eq:err_bd1}
\|P_{K}x-\hat{x}\|_{1,\bom} &\leq C \sigma_{\s, \M}(P_K x)_{1,\bom} 
                            + D \frac{\sqrt{r}}{\sqrt{\theta}}(\eta + \eta') \\
\|P_{K}x-\hat{x}\|_{2} &\leq (1+r^{1/4}) \( C \frac{\sigma_{\s,\M}(P_{K}x)_{1,\bom}}{\sqrt{r}} + D\frac{1}{\sqrt{\theta}}(\eta  + \eta')\)
\label{eq:err_bd2}
\end{align}
where $C = 2(2+\sqrt{3})/(2-\sqrt{3})$ and 
$D = 8\sqrt{2}/(2-\sqrt{3})$.
\end{theorem}

\begin{figure}
\begin{center}
 {\small
\begin{tabular}{@{\hspace{0pt}}c@{\hspace{6pt}}c@{\hspace{6pt}}c@{\hspace{0pt}}}
\includegraphics[height=5.5cm]{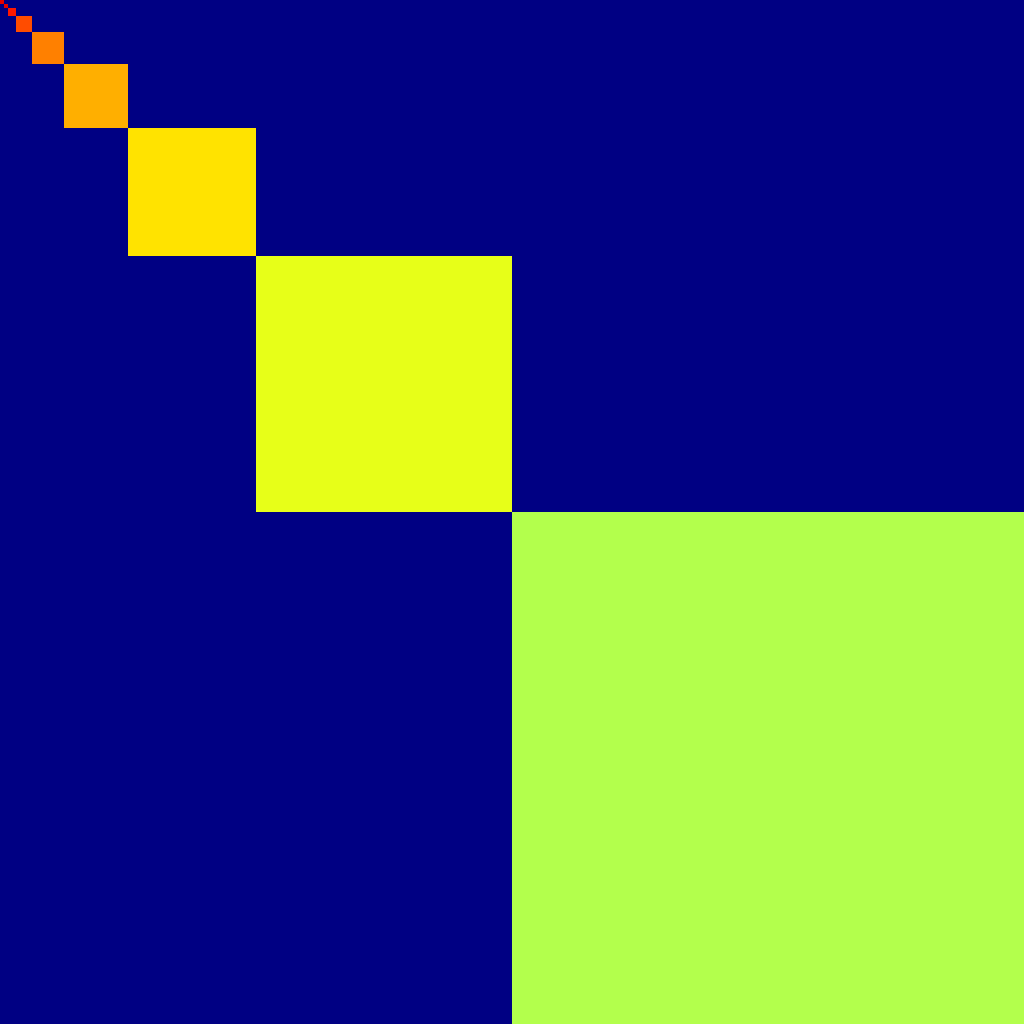}
&
\includegraphics[height=5.5cm]{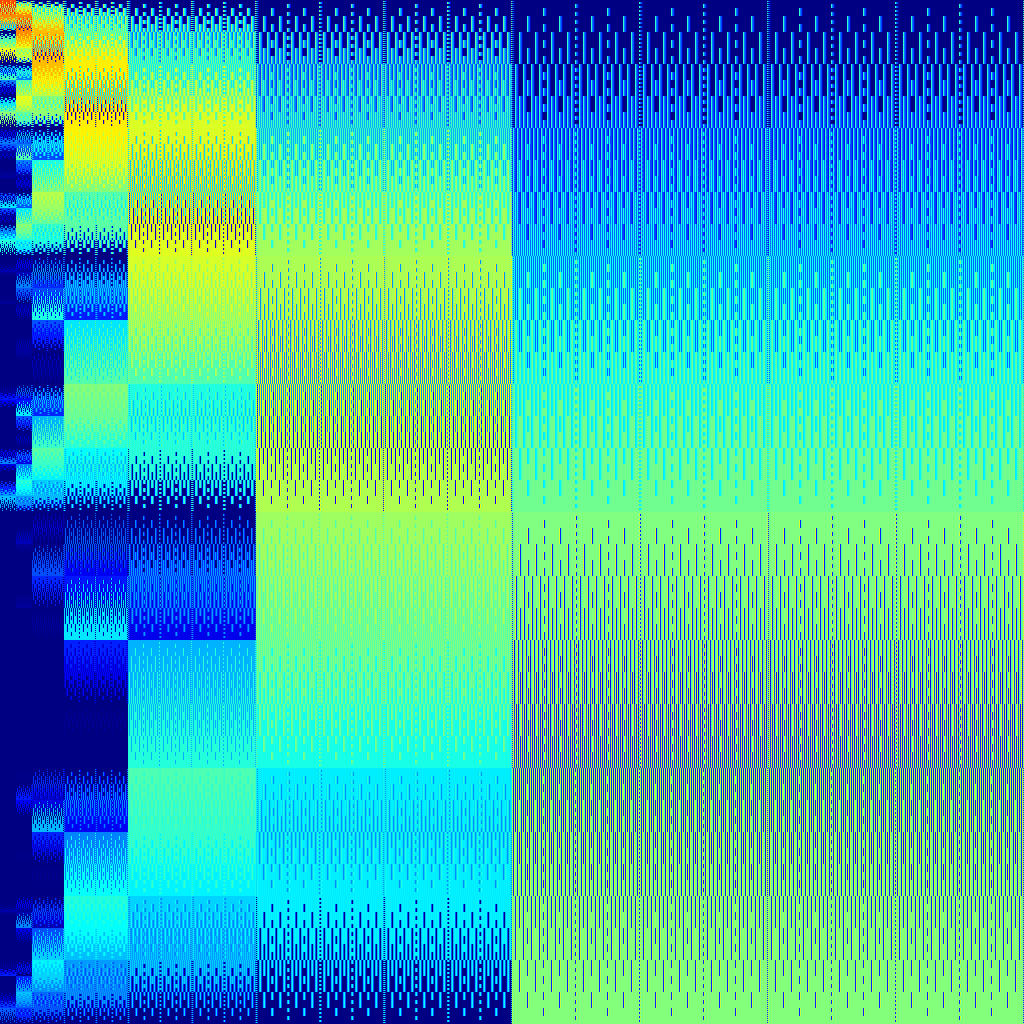} 
&
\includegraphics[height=5.5cm]{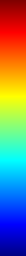}
\\
Haar (DB1)& DB4
\end{tabular}
}
\caption{The absolute values in log scale of the matrix $P_{M}UP_{M}$ for $U = [\Bwh, \Bwave^{J_0,\nu}]$, with $\nu = 1$ (left) and $\nu=4$ (middle). The rightmost image is the colorbar.}
\label{fig:U_Haar_DB4}
\end{center}
\end{figure}

\begin{remark}\label{r:K_issue_fudge}
Note that the second condition (ii) can be guaranteed using Proposition \ref{prop:norm12}.  Indeed, it suffices for $K$ to satisfy
\[
\frac{\nm{P^{\perp}_K x}_{1}}{\sqrt{K}} \lesssim \frac{\eta'}{\sqrt{N}}.
\]
Hence, given any \textit{a priori} estimates on the decay of the coefficients $x$ (such as in the case of wavelets), one can use this to determine a suitable $K$.
\end{remark}


\subsection{Uniform recovery for Haar wavelets}
\label{subsec:rec_Haar}

Below we shall see that for the Haar wavelet, $P_{N}UP_{N}$ will be an isometry
for $N = 2^r$ where $r \in \N$.  This can also be seen from Figure
\ref{fig:U_Haar_DB4}, where $U = [\Bwh, \Bwave^{J_0,\nu}]$ is perfectly
block diagonal for $\nu=1$.  This means that the G-RIPL, reduces to
the $I$-adjusted RIPL, or simply the RIPL, which we know from the finite
dimensional case.  Notice in particular that we also avoid any considerations
where $K > M=N$ as above, since $HP_{M}^{\perp} = 0$.

\begin{proposition}
    \label{lem:Haar_isometry}
Let $U = [\Bwh, \Bwave^{J_0, 1}]$ and let $N = 2^{k}$, for some $k\in \N$ with 
$k \geq J_0+1$. Then $P_{N}UP_{N}$ is an isometry on $\C^N$.   
\end{proposition}
\begin{proposition}
    \label{lem:1d_coher_bound_haar}
    Let $U = [\Bwh, \Bwave^{J_0, 1}]$ and let 
    $ \M = \Nm = [2^{J_0+1}, \ldots, 2^{J_0+r}] $    be sparsity and sampling levels, respectively.
Then the local coherences of $U$ are 
\[  \mu_{kl} = \begin{cases} 2^{-J_0-k+1} &\text{if } k = l \\ 
0 &\text{if } k \neq l
\end{cases}    \]
\end{proposition}
It is now straightforward to derive the following:

\begin{theorem}
    \label{thm:haar1}
    Let $U = [\Bwh, \Bwave^{J_0, 1}]$ and let 
    $\M = \Nm = [2^{J_0+1}, \ldots, 2^{J_0+r}] $
    be sparsity and sampling levels. Let $s \in \N^{r}$ be local sparsities 
    and $\m\in \N^r$ be local sampling densities. 
    Let $\eps,\delta \in (0,1)$ and $0\leq r_0 \leq r$. Let $\tilde{m} = m_{r_0+1} + \ldots + m_r$ and $s = s_1 + \ldots + s_r$. Suppose that the $m_k$'s satisfies $m_{k} =
    N_{k}-N_{k-1}$ for $k = 1, \ldots, r_0$ and 
    \begin{equation}
        m_{k} \gtrsim \delta^{-2} s_k \( r\log(2\tilde{m})\log(2N) \log^2(2s) +
        \log(\eps^{-1})\), \quad\text{ for } k = r_0+1,\ldots, r.
    \end{equation}
Then with 
probability $1-\eps$ the matrix \eqref{def:A_infinite} satisfies the RIPL 
with constant $\delta_{\s,\M}\leq \delta$.  
\end{theorem}
\begin{proof}
Using Proposition \ref{lem:Haar_isometry} 
we know that $P_{N}UP_{N}$ is an isometry. Thus inserting the  
local coherences from Proposition \ref{lem:1d_coher_bound_haar}
into \eqref{eq:Li_RIPL_mk} in Theorem \ref{thm:Li_RIPL_mk} gives to the result.
\end{proof}

\begin{theorem}
    \label{thm:haar2}
Let $U = [\Bwh, \Bwave^{J_0, 1}]$ and let $\M = \Nm = [2^{J_0+1}, \ldots, 2^{J_0+r}]$ 
be sparsity and sampling levels. Let $\s \in \N^{r}$ be local sparsities,
$\bom = (s_{1}^{1/2}, \ldots, s_{r}^{1/2})$ be weights  and $\m\in \N^r$ be
local sampling densities. Let $\eps \in (0,1)$ and let
$0 \leq r_0 \leq r$. Let $m=m_1 +\ldots + m_r$, $\tilde{m}= m_{r_0+1}+ \cdots+m_r$ and $s = s_1 +
\ldots + s_r$. Suppose we sample $m_{k} = N_{k} - N_{k-1}$ for $k=1,\ldots,
r_0$ and 
    \begin{equation*}
    m_k \gtrsim r \cdot s_k
     \cdot \left (  r\log(2 \tilde{m})
    \log(2N) \log^2(2s) + \log(\epsilon^{-1}) \right ),
    \end{equation*}
    for $k = r_0+1,\ldots, r$. Let $H \in \C^{m\times \infty}$ be as in
\eqref{eq:H_mat} with $A = HP_{M}$. Let $x \in \ell^2(\N)$ and $e \in \C^m$
with $\|e\|_2\leq \eta$ for some $\eta \geq 0$. Set $\tilde{y} = Ax + e$. Then
any solution $\hat{x}$ of the optimization problem
    \[ \minimize_{z \in \C^{M}} \|z\|_{1,\bom} \quad\text{subject to} \quad 
       \|Az - \tilde{y}\|_2 \leq \eta \]
    satisfies
    \begin{align*}
    \|P_{M}x-\hat{x}\|_{1,\bom} &\leq C \sigma_{\s, \M}(P_{M}x)_{1,\bom} + D \sqrt{r}
                \eta \\
    \|P_{M}x-\hat{x}\|_{2} &\leq (1+r^{1/4}) \( C \frac{\sigma_{\s,\M}(P_{M}x)_{1,\bom}}{\sqrt{r}} + D\eta \)
    \end{align*}
    with probability $1-\eps$, where $C = 2(2+\sqrt{3})/(2-\sqrt{3})$ and $D =
    8\sqrt{2}/(2-\sqrt{3})$.
\end{theorem}
\begin{proof}
    Proposition \ref{lem:Haar_isometry} gives $G = \sqrt{P_{M}U^*P_N U P_M} = \sqrt{I}=I$. Next notice that $\Sboms =r$ and that $P_{M}x \in \{z \in \C^M : \|Az -
    \tilde{y}\|_2 \leq \eta\}$ since $\|HP_{M}^{\perp}\| = 0$. Using Theorem 
    \ref{thm:GRIPL_imply_uniform} we see that we can guarantee recovery of
$(\s,\M)$-sparse vectors, if $A$ satisfies the RIPL with constant
$\delta_{\t,\M} \leq 1/2$, where $t_l = \min \{M_{l}-M_{l-1}, 8rs_l\}$. Using Theorem 
\ref{thm:haar1} gives the result.    
\end{proof}

\section{Proof of results in Section \ref{sec:infin_dim}}
\label{sec:proof1}

When deriving uniform recovery guarantees via the RIP, it is typical to proceed
as follows.  First, one shows that the RIP implies the so-called \emph{robust
Null space Property (rNSP)} of order $s$ (see Def.\ 4.17 in \cite{Foucart13}).
Second, one the shows that the rNSP implies stable and robust recovery.
Thus the line of implications reads
\[ \text{(RIP) } \implies \text{ (rNSP)} \implies \text{(uniform recovery)}. \]
A similar line of implications holds for the RIPL and the corresponding \textit{robust Null Space
Property in levels (rNSPL)}; see Def.\ 3.6 in \cite{Bastounis17}).

Both of the recovery guarantees for matrices satisfying the rNSP and rNSPL
consider minimizers of the unweighed quadratically-constrained basis pursuit
(QCBP) optimization problem. In our setup we consider minimizers of the
weighted QCBP. We have therefore generalized the rNSPL to what we call the
weighted robust null space property in levels.

For the sufficient condition for the G-RIPL in Theorem
\ref{t:RIPlevels1}, the proof follows along similar lines as in \cite{Li17}. We
only sketch the main differences here.  

\subsection{The weighted rNSPL and norm bounds}

For a set $\Theta \subseteq \{1,\ldots,
M\}$ and a vector $x \in \C^{M}$ we let the vector $x_{\Theta}$ be given by
\[ (x_{\Theta})_i = \begin{cases} x_i & i \in \Theta \\ 0 & i \not\in\Theta
\end{cases}. \]
We also define 
\[
E_{\s,\M} = \{ \Theta \subseteq \{1,\ldots,M\} : |\Theta \cap \{M_{l-1}+1,\ldots, M_l\}| \leq s_l,\text{ for } l =1,\ldots,r\}.
\] 
\begin{definition}[weigthed rNSP in levels]
Let $\M,\s \in \N^r$ be sparsity levels and local sparsities, respectively.
For positive weights $\bom \in \R^{r+1}$, we say that $A \in \mathbb{C}^{m
\times M}$ satisfies the \emph{weighted robust Null Space Property in
Levels} (weighted rNSPL) of order $(\s, \M)$ with constants $0 < \rho <1$ and
$\gamma> 0$ if 
\begin{equation}
    \label{eq:omrNSPL}
    \|x_{\Theta}\|_2 
    \leq \frac{\rho \| x_{\Theta^{c}}\|_{1,\bom}}{\sqrt{\Sboms}}
    + \gamma \|A x\|_2
\end{equation}
for all $x \in \C^M$ and all $\Theta \in E_{\s,\M}$. 
\end{definition}

\begin{lemma}[weighted rNSPL implies $\ell^{(1,\bom)}$-distance bound]
\label{l:rNSPLdistance1w}
    Suppose that $A \in \C^{m\times M}$ satisfies the weighted rNSPL of
order $(\s,\M)$ with constants $0 < \rho < 1$ and $\gamma > 0$. Let $x,z \in
\C^{M}$. Then 
\begin{equation}
\| z-x\|_{1,\bom} \leq \frac{1+\rho}{1-\rho}\( 2\sigma_{\s,\M} (x)_{1,\bom} 
+ \|z\|_{1,\bom} - \|x\|_{1,\bom} \) + \frac{2\gamma}{1-\rho} \sqrt{\Sboms}
\|A(z-x)\|_{2}.
\end{equation}
\end{lemma}
\begin{proof}
    Let $v=z-x$ and $\Theta \in E_{\s,\M}$ be such that $\|x_{\Theta^{c}}\|_{1,\bom} 
= \sigma_{\s,\M}(x)_{1,\bom}$. Then 
\begin{align*}
\nm{x}_{1,\bom} + \nm{v_{\Theta^c}}_{1,\bom}& \leq 2 \nm{x_{\Theta^c}}_{1,\bom} + \nm{x_{\Theta}}_{1,\bom} + \nm{z_{\Theta^c}}_{1,\bom}
\\
& = 2 \nm{x_{\Theta^c}}_{1,\bom} + \nm{x_{\Theta}}_{1,\bom} + \nm{z}_{1,\bom} - \nm{z_{\Theta}}_{1,\bom}
\\
& \leq 2 \sigma_{\mb{s},\mb{M}}(x)_{1,\bom} + \nm{v_{\Theta}}_{1,\bom} + \nm{z}_{1,\bom},
\end{align*}
which implies that
\be{
\label{frost}
\nm{v_{\Theta^c}}_{1,\bom} \leq 2 \sigma_{\mb{s},\mb{M}}(x)_{1,\bom}  + \nm{z}_{1,\bom} - \nm{x}_{1,\bom} + \nm{v_{\Theta}}_{1,\bom}.
}
Now consider $\nm{v_{\Theta}}_{1,\bom}$.  By the weighted rNSPL, we have 
\bes{
\nm{v_{\Theta}}_{1,\bom} \leq \sqrt{\Sboms} \nm{v_{\Theta}}_{2} \leq \rho \nm{v_{\Theta^c}}_{1,\bom} + \sqrt{\Sboms} \gamma \nm{A v }_{2}.
}
Hence \eqref{frost} gives
\bes{
\nm{v_{\Theta}}_{1,\bom} \leq \rho \left ( 2 \sigma_{\mb{s},\mb{M}}(x)_{1,\bom}  + \nm{z}_{1,\bom} - \nm{x}_{1,\bom} + \nm{v_{\Theta}}_{1,\bom} \right ) + \sqrt{\Sboms} \gamma \nm{A v}_{2},
}
and after rearranging we get
\bes{
\nm{v_{\Theta}}_{1,\bom} \leq \frac{\rho}{1-\rho} \left ( 2 \sigma_{\mb{s},\mb{M}}(x)_{1,\bom}  + \nm{z}_{1,\bom} - \nm{x}_{1,\bom}  \right ) + \frac{\gamma}{1-\rho} \sqrt{\Sboms} \nm{A v }_{2}.
}
Therefore, using this and \eqref{frost} once more, we deduce that
\eas{
\nm{z-x}_{1,\bom} &= \nm{v_{\Theta}}_{1,\bom} + \nm{v_{\Theta^c}}_{1,\bom}
\\
& \leq  2 \nm{v_{\Theta}}_{1,\bom} + \left ( 2 \sigma_{\mb{s},\mb{M}}(x)_{1,\bom}  + \nm{z}_{1,\bom} - \nm{x}_{1,\bom}  \right )
\\
& \leq \frac{1+\rho}{1-\rho}  \left ( 2 \sigma_{\mb{s},\mb{M}}(x)_{1,\bom}  + \nm{z}_{1,\bom} - \nm{x}_{1,\bom}  \right ) + \frac{2\gamma}{1-\rho} \sqrt{\Sboms} \nm{A (z-x)}_{2},
}
which gives the result.
\end{proof}

\begin{lemma}[weighted rNSPL implies $\ell^2$ distance bound]
\label{l:rNSPLdistance2}
Suppose that $A \in \C^{m\times M}$ satisfies the weighted rNSPL of order $(\mb{s},\mb{M})$ with constants $0 < \rho < 1$ and $\gamma > 0$.  Let $x,z \in \C^M$.  Then
\begin{equation}
\nm{z-x}_{2} \leq \left ( \rho + (1+\rho)(\Sboms/\zetas)^{1/4}/2\right) \frac{\nm{z-x}_{1,\bom}}{\sqrt{\Sboms}} + \left ( 1 + (\Sboms/\zetas)^{1/4}/2\right) \gamma \nm{A (z-x)}_{2}.
\end{equation}
\end{lemma}

\begin{proof}
Let $v = z-x$ and $\Theta = \Theta_1\cup \cdots \cup \Theta_r$, where $\Theta_l \subseteq \{M_{l-1}+1,\ldots,M_l\}$, $| \Theta_l | = s_l$ is the index set of the largest $s_l$ coefficients of $P^{M_{l-1}}_{M_l} v$ in absolute value.  Then
\bes{
\nm{v_{\Theta_l}}_{2} = \sqrt{\sum_{i \in \Theta_l} |v_i |^2} \geq \sqrt{s_l} \min_{i \in \Theta_l} |v_i | \geq \sqrt{s_l} \max_{\substack{M_{l-1} < i \leq M_l \\ i \notin \Theta_l}} |v_i |,\quad l=1,\ldots,r,
}
which gives
\eas{
\nm{v_{\Theta^c}}^2_{2} &= \sum^{r}_{l=1} \sum_{\substack{M_{l-1} < i \leq M_l \\ i \notin \Theta_l}} |v_i|^2  \leq \sum^{r}_{l=1} \max_{\substack{M_{l-1} < i \leq M_l \\ i \notin \Theta_l}} |v_i | \sum_{\substack{M_{l-1} < i \leq M_l \\ i \notin \Theta_l}} |v_i|
\\
& \leq \sum^{r}_{l=1} \frac{\nm{v_{\Theta_l}}_{2}}{\sqrt{s_l}} \sum_{\substack{M_{l-1} < i \leq M_l \\ i \notin \Theta_l}} |v_i|
\leq \max_{l=1,\ldots,r} \left \{\frac{\nm{v_{\Theta_l}}_{2}}{\om_{l}\sqrt{s_l}} \right \} \sum^{r}_{l=1} \om_{l} \sum_{\substack{M_{l-1} < i \leq M_l \\ i \notin \Theta_l}} |v_i|
\\
& \leq \max_{l=1,\ldots,r} \left \{\frac{\nm{v_{\Theta_l}}_{2}}{\om_{l}\sqrt{s_l}} \right \} \nm{v_{\Theta^c} }_{1,\bom} 
}
Since $\nm{v_{\Theta_l}}_{2} \leq \nm{v_{\Theta}}_{2}$ we deduce that
\bes{
\nm{v_{\Theta^c}}_{2} \leq \sqrt{\frac{\nm{v_{\Theta}}_{2}\nm{v_{\Theta^c} }_{1,\bom}}{\min_{l=1,\ldots,r} \{ \om_{l} \sqrt{s_l} \} } } = \sqrt{\frac{\nm{v_{\Theta}}_{2}\nm{v_{\Theta^c} }_{1,\bom}}{\sqrt{\zetas}}}.
}
Applying Young's inequality $a b \leq \frac12 a^2 + \frac12 b^2$, we obtain
\bes{
\nm{v_{\Theta^c}}_{2} \leq \frac{(\Sboms / \zetas)^{1/4} }{2} \frac{\nm{v_{\Theta^c}}_{1,\bom}}{\sqrt{\Sboms}}+ \frac{(\Sboms / \zetas)^{1/4}}{2} \nm{v_{\Theta}}_{2}.
}
Hence
\bes{
\nm{v}_{2} \leq \nm{v_{\Theta}}_{2} + \nm{v_{\Theta^c}}_{2} \leq \left ( 1 + (\Sboms / \zetas)^{1/4}/2\right)  \nm{v_{\Theta}}_{2} +  \frac{(\Sboms / \zetas)^{1/4} }{2} \frac{\nm{v_{\Theta^c}}_{1,\bom}}{\sqrt{\Sboms}}.
}
We now use the weighted rNSPL to get
\bes{
\nm{v}_{2} \leq  \left ( \rho + (1+\rho)(\Sboms / \zetas)^{1/4}/2\right) \frac{\nm{v_{\Theta^c}}_{1,\bom}}{\sqrt{\Sboms}} + \left ( 1 + (\Sboms / \zetas)^{1/4}/2\right) \gamma \nm{A v}_{2}.
}
To complete the proof, we use the inequality $\nm{v_{\Theta^c}}_{1,\bom} \leq \nm{v}_{1,\bom}$.
\end{proof}

\subsection{Weighted rNSPL implies uniform recovery}

\begin{theorem} \label{lem:min_om_bound}
Let $\M,\s \in \N^r$ be sparsity levels and local sparsities, respectively, 
and let $\bom \in \R^{r+1}$ be positive weights.
Let $x \in \C^{K}$, with $K>M$ and $e \in \C^m$ with
$\|e\|_{2}\leq \eta$.  Set $y = Ax+e$. Let $A \in \C^{m \times K}$
and suppose that $AP_{M}$ satisfies the weighted rNSP in levels of order
$(\s,\M)$ with constants $\rho = \sqrt{3}/2$ and $\gamma > 0$.
If  
\begin{equation}
    \label{eq:om_cond}
 \om_{r+1} \geq \sqrt{\Sboms}\( \frac{1}{3(1+(\Sboms/\zetas)^{1/4})} 
                + 2 \gamma \|AP_{K}^{M}\|_{1\to 2}  \)
\end{equation}
then any solution $\hat{x}$ of the optimization problem
     \begin{align}
        \label{eq:QCBP123}
        \minimize_{z \in \C^K} \|z\|_{1, \bs{\om}} \quad \text{subject to} \quad
        \|Az- y\|_2 \leq \eta 
     \end{align}
satisfies
\begin{align*}
  \|x-\hat{x}\|_{1,\bs{\om}} \leq& C
  \sigma_{\s,\M}(x)_{1,\bom} + D \gamma  \sqrt{\Sboms} \eta\\
  \|x-\hat{x}\|_{2}  \leq& \(1+ (\Sboms/\zetas)^{1/4} \) \( C
\frac{\sigma_{\s,\M}(x)_{1,\bom}}{\sqrt{\Sboms}} +D \gamma\eta \),
\end{align*}
where $C = 2(2+\sqrt{3})/(2-\sqrt{3})$ and $D = 8/(2-\sqrt{3})$.
\end{theorem}
{\allowdisplaybreaks
\begin{proof}
Recall that $\rho = \sqrt{3}/2$, and notice that this gives $C/2 = (1+\rho)/(1-\rho)$ and $D/2 = 2/(1-\rho)$. Next we consider the bound \eqref{eq:om_cond}, and  note that  
this bound implies 
\begin{align}
    \om_{r+1} &\geq \gamma \sqrt{\Sboms} \|AP_{K}^{M}\|_{1\to 2} / \rho \\ 
   1 + 2\rho &\geq 1 + 2\gamma \sqrt{\Sboms} \|AP_{K}^{M}\|_{1\to 2}/ \om_{r+1} \\
   1 + \rho &\geq 1-\rho + 2\gamma \sqrt{\Sboms} \|AP_{K}^{M}\|_{1\to 2}/ \om_{r+1} \\
   \frac{C}{2}&\geq 1 + \frac{D}{2} \gamma \sqrt{\Sboms}
            \|AP_{K}^{M}\|_{1\to 2} /\om_{r+1}.
    \label{eq:cond_om2}
\end{align}

We also note that \eqref{eq:om_cond} implies
\begin{align*}
    \om_{r+1} \geq&  \( \frac{1}{3(1+(\Sboms/\zetas)^{1/4})} 
                + 2 \gamma \|AP_{M}^{\perp}\|_{1\to 2}  \) \sqrt{\Sboms}\\
    \geq& \( \frac{2}{C(1+(\Sboms/\zetas)^{1/4})}
            +\frac{D}{C} \gamma \|AP_{K}^{M}\|_{1\to 2} \)\sqrt{\Sboms}
\end{align*}
which can be written as
\begin{align}
(1+(\Sboms/\zetas)^{1/4})(C/2)\frac{1}{\sqrt{\Sboms}} \geq \((D/2)(1+(\Sboms/\zetas)^{1/4}) \gamma \|AP_{K}^{M}\|_{1\to 2} + 1\)/\om_{r+1}.
\label{eq:om2_bound_rewritten}
\end{align}
Next set $v = x - \hat{x}$ and consider the $\ell^{(1,\bom)}$-bound. 
    First notice that since $AP_{M}$ satisfies the weighted rNSPL, Lemma
\ref{l:rNSPLdistance1w} gives  
\begin{equation}
    \label{eq:bound_vM}
    \begin{split}
        \|P_{M} v\|_{1,\bom} \leq& (C/2)\( 2\sigma_{\s,\M} (P_M x)_{1,\bom} 
+ \|P_{M}\hat{x}\|_{1,\bom} - \|P_{M}x\|_{1,\bom} \)  
+ (D/2)\gamma\sqrt{\Sboms}
\|AP_M v\|_{2} .
    \end{split}
\end{equation}
Here the last term can be bounded by 
\begin{align}
    \|AP_{M}v\|_2 &\leq \|Av+y-y\|_{2} + \|AP_{K}^{M} v\|_{2}
\leq 2\eta + \frac{\|AP_{K}^{M}\|_{1\to 2}}{\om_{r+1}} 
\|P_{K}^{M}v\|_{1,\bom} \\
    &\leq 2\eta + \frac{\|AP_{K}^{M}\|_{1\to 2}}{\om_{r+1}} 
\(\|P_{K}^{M}x\|_{1,\bom} + \|P_{K}^{M}\hat{x}\|_{1,\bom} \),
\label{eq:bound_AMv}
\end{align}
since both $x$ and $\hat{x}$ are feasible.
Combining \eqref{eq:cond_om2}, \eqref{eq:bound_vM} and \eqref{eq:bound_AMv} gives
\begin{align*}
\|v\|_{1,\bom} \leq& \|P_{M} v\|_{1,\bom} + \|P_{K}^{M}x\|_{1,\bom} +
\|P_{K}^{M}\hat{x}\|_{1,\bom} \\
\leq& (C/2) \big( 2\sigma_{\s,\M} (P_M x)_{1,\bom} 
+ \|P_{M}\hat{x}\|_{1,\bom} - \|P_{M}x\|_{1,\bom} \big) + \|P_{K}^{M}x\|_{1,\bom} +
      \|P_{K}^{M}\hat{x}\|_{1,\bom} \\
       &+ (D/2)\gamma \sqrt{\Sboms} \|AP_{M}v\|_2 \\
\leq& (C/2) \big( 2\sigma_{\s,\M} (P_M x)_{1,\bom} 
+ \|P_{M}\hat{x}\|_{1,\bom} - \|P_{M}x\|_{1,\bom} \big)
+ D\gamma \sqrt{\Sboms} \eta  \\
   &+ \(1 + (D/2)\gamma \sqrt{\Sboms} 
  \frac{\|AP_{K}^{M}\|_{1\to 2}}{\om_{r+1}}\) \(\|P_{K}^{M}x\|_{1,\bom}
    + \|P_{K}^{M}\hat{x}\|_{1,\bom} \) 
    \\ 
\leq& (C/2) \big( 2\sigma_{\s,\M} ( x)_{1,\bom} 
+ \|\hat{x}\|_{1,\bom} - \|x\|_{1,\bom} \big) 
    + D\gamma \sqrt{\Sboms} \eta.  
\end{align*}
Using that $\hat{x}$ is a minimizer of \eqref{eq:QCBP123} gives the desired 
bound. 

We now consider the $\ell^2$-bound. First note that 
\begin{equation}
    \label{eq:vl2_simple}
    \|v\|_{2} \leq \|P_{M}v\|_{2} + \|P_{K}^{M}v\|_{2} \leq \|P_{M}v\|_{2} + 
    \frac{1}{\om_{r+1}}\|P_{K}^{M}v\|_{1,\bom}.
\end{equation}
We shall also need  
\begin{equation}
\begin{split}
    &(\rho +(1+\rho)(\Sboms/\zetas)^{1/4}/2)\frac{2}{1-\rho} + (1+(\Sboms/\zetas)^{1/4}/2)\\
=& (D/4)\big(2\rho +(1+\rho)(\Sboms/\zetas)^{1/4} + (1-\rho) 
+(1-\rho)(\Sboms/\zetas)^{1/4}/2\big) \\
=& (D/4)\big((1+\rho) + \tfrac{1}{2}(3+\rho) (\Sboms/\zetas)^{1/4} \big) \\
\leq& (D/2) \big( 1+ (\Sboms/\zetas)^{1/4}\big)  
\end{split}
\label{eq:long_ineq}
\end{equation}
Again, since $AP_{M}$ satisfies the 
weighted rNSPL we can apply Lemma \ref{l:rNSPLdistance2}, Lemma
\ref{l:rNSPLdistance1w} and inequality \eqref{eq:long_ineq} to obtain the bound 
\begin{equation}
\begin{split}
    \|P_{M}v\|_{2} \leq& 
\( \rho +(1+\rho)(\Sboms/\zetas)^{1/4}/2 \) \frac{\nm{P_{M}v}_{1,\bom}}{\sqrt{\Sboms}} 
+ \left (1 + (\Sboms/\zetas)^{1/4}/2\right) \gamma \nm{AP_{M}v}_{2} \\
\leq&
\( 1 + (\Sboms/\zetas)^{1/4} \) 
\(C/2\)
\frac{2\sigma_{\s,\M}(P_{M}x)_{1,\bom} 
      + \|P_{M}\hat{x}\|_{1,\bom} 
        - \|P_{M}x\|_{1,\bom}}{\sqrt{\Sboms}} \\
&+ \(\rho + (1+\rho)(\Sboms/\zetas)^{1/4}/2\) \frac{2\gamma}{1-\rho} \|AP_{M}v\|_{2} \\
&+ \left ( 1 + (\Sboms/\zetas)^{1/4}/2\right) \gamma \nm{AP_{M}v}_{2} \\ 
\leq&
\( 1 + (\Sboms/\zetas)^{1/4} \) 
\(C/2\)
\frac{2\sigma_{\s,\M}(P_{M}x)_{1,\bom} 
      + \|P_{M}\hat{x}\|_{1,\bom} 
        - \|P_{M}x\|_{1,\bom}}{\sqrt{\Sboms}} \\
+&
 (D/2)\(1+(\Sboms/\zetas)^{1/4} \)
\gamma \nm{AP_{M}v}_{2}.
\end{split}
\label{eq:P_Mv2}
\end{equation}
Combining \eqref{eq:om2_bound_rewritten}, \eqref{eq:bound_AMv}, \eqref{eq:vl2_simple}, \eqref{eq:P_Mv2} and  now gives 
\begin{align*}
\|v\|_{2} \leq&
\(1 + (\Sboms/\zetas)^{1/4} \)\(C/2\) 
\frac{2\sigma_{\s,\M}(P_{M}x)_{1,\bom} 
      + \|P_{M}\hat{x}\|_{1,\bom} 
        - \|P_{M}x\|_{1,\bom}}{\sqrt{\Sboms}} \\
   &+(D/2)\(1+(\Sboms/\zetas)^{1/4} \)
\gamma \nm{AP_{M}v}_{2} + \frac{1}{\om_{r+1}} \|P_{K}^{M}v\|_{1,\bom} \\
\leq& 
\( 1 + (\Sboms/\zetas)^{1/4} \)\(C/2\) 
\frac{2\sigma_{\s,\M}(P_{M}x)_{1,\bom} 
      + \|P_{M}\hat{x}\|_{1,\bom} 
        - \|P_{M}x\|_{1,\bom}}{\sqrt{\Sboms}} \\
&+\((D/2)\(1+(\Sboms/\zetas)^{1/4} \) \gamma
\nm{AP_{K}^{M}}_{1 \to 2} +1\) \frac{\|P_{K}^{M}x\|_{1,\bom} +\|P_{K}^{M}\hat{x}\|_{1,\bom}}{\om_{r+1}}\\
   &+\(1+(\Sboms/\zetas)^{1/4} \) D\gamma\eta \\
\leq& 
\( 1 +(\Sboms/\zetas)^{1/4} \)\(C/2\) 
\frac{2\sigma_{\s,\M}(x)_{1,\bom} 
      + \|\hat{x}\|_{1,\bom} 
        - \|P_{K}x\|_{1,\bom}}{\sqrt{\Sboms}} \\
   &+\(1+(\Sboms/\zetas)^{1/4} \) D\gamma\eta
\end{align*}
Using that $\hat{x}$ is a minimizer of \eqref{eq:QCBP123} completes the proof. 
\end{proof} }

\subsection{G-RIPL implies weighted rNSPL}
\begin{theorem}
\label{thm:GRIPLimpliesrNSP}
Let $A \in \mathbb{C}^{m \times M}$ and let $G \in \mathbb{C}^{M \times
M}$be invertible. 
Let $\M\in \N^r$ be sparsity levels, $\s,\t\in \N^r$ be local sparsities and let $\bom \in \R^{r}$ be positive weights.   
Suppose that $A$ satisfies the G-RIPL of order
$(\t,\M)$ with constant $0 < \delta_{\t,\M} < 1$, where  
\begin{equation}
\label{eq:t_kcond}
t_{l} = \min\left\{ M_{l}-M_{l-1},  2
\ceil{\left( \frac{1+\delta}{1-\delta} \right)\frac{\kappa(G)^2}{\rho^2\om^{2}_{l}}  \Sboms}\right\}, 
\quad \text{for } l = 1, \ldots, r.
\end{equation}
Then $A$ satisfies the weighted rNSP in levels of order $(\s,\M)$ with constants
$0 < \rho < 1$ and $\gamma = \sqrt{2} \| G^{-1} \|_{2}$.
\end{theorem}

\begin{proof}
Let $x \in \mathbb{C}^{K}$ be such that $P_{M}^{\perp} x = 0$ and let 
$\Theta = \Theta_1 \cup \cdots \cup \Theta_r$, where $\Theta_{l}$ is the set of
the largest $s_l$ indices of $P^{M_{l-1}}_{M_l} x$ in absolute value.  
If $t_l = M_{l}-M_{l-1}$, let $T_{l,0} = \{M_{l-1}+1, \ldots, M_l\}$ and let 
$T_{l,k} = \emptyset$ for $k \geq 1$. For $t_l < M_{l}-M_{l-1}$ let 
$T_{l,0}$ be the
index set of the largest  $t_l/2$ values of $|P^{M_{l-1}}_{M_l} x|$, and let
$T_{l,1}$ be the index set of the next $t_l/2$ largest values and so forth.  In
the case where there are less than $t_l/2$ values left at iteration $k$, we let
$T_{l,k}$ be the remaining indices.  Let $T_{k} = T_{1,k} \cup \cdots \cup
T_{r,k}$ and let   $T_{\{0,1\}} = T_{0}\cup T_{1}$.  Since $\Theta \subseteq T_{\{0,1\}}$ we
have 
\begin{equation*}
\| x_{\Theta} \|^2_{2}\leq \| x_{T_{\{0,1\}}} \|^2_2 \leq \| G^{-1} \|^2_{2} \| G
x_{T_{\{0,1\}}} \|^2_{2}  \leq \frac{ \| G^{-1} \|^2_{2}}{1-\delta}  \| A x_{T_{\{0,1\}}}
\|^2_{2} 
\end{equation*}
where $\delta = \delta_{\t,\mb{M}}$.  Note that
\begin{equation*}
A x_{T_{\{0,1\}}} = A x - \sum_{k \geq 2} A x_{T_{k}},
\end{equation*}
Then
\begin{align*}
\| A x_{T_{\{0,1\}}} \|^2_{2} &= \ip{A x_{T_{\{0,1\}}}}{A x}
-  \sum_{k \geq 2} \ip{A
x_{T_{\{0,1\}}}}{A x_{T_{k}}} \\
& \leq \| A x_{T_{\{0,1\}}} \|_2 \| A x \|_{2} + \| A x_{T_{\{0,1\}}} \|_2
 \sum_{k \geq 2} \| A x_{T_{k}} \|_2 \\
& \leq \| A x_{T_{\{0,1\}}} \|_2 \| A x \|_{2} + \sqrt{1+\delta} 
\|A x_{T_{\{0,1\}}} \|_2  \sum_{k \geq 2} \| G x_{T_{k}} \|_2 \\
 &\leq \| A x_{T_{\{0,1\}}} \|_2 \| A x \|_{2} + \sqrt{1+\delta} \| G \|_{2} 
 \| A x_{T_{\{0,1\}}} \|_2  \sum_{k \geq 2} \| x_{T_{k}} \|_2
\end{align*}
Set $\Delta = \{ l\in \{1,\ldots,r\}: t_{l} < M_{l}-M_{l-1} \}$ and notice 
that $T_{l,k} = \emptyset$ for $l  \in \{1, \ldots, r\}\setminus \Delta$ and $k \geq 1$.
Thus for $k \geq 2$ we get 
\begin{align*}
\| x_{T_k} \|_{2}^{2} &= \sum_{l\in \Delta} \| x_{T_{l,k}} \|^2_2  \leq
\sum_{l\in \Delta} \frac{2\|x_{T_{l,k-1}}\|^2_{1}}{t_l}
= \sum_{l\in\Delta} \frac{2\|x_{T_{l,k-1}}\|^2_{1}\om_{l}^{2}}{t_l\om_{l}^{2}} \\ 
&\leq \frac{ \sum_{l\in\Delta} 2\| x_{T_{l,k-1}} \|_{1,\bom}^{2}
}{\min_{l\in \Delta} \{ \om_{l}^{2}t_l \} }  \leq \frac{2\|
x_{T_{k-1}} \|_{1,\bom}^{2}}{\min_{l\in \Delta} \{
\om_{l}^{2}t_l \} } .
\end{align*}
Therefore
\begin{align*}
\| A x_{T_{\{0,1\}}} \|_{2} & \leq \| A x \|_{2} + \frac{\sqrt{2(1+\delta)} \| G
\|_{2}}{\sqrt{\min_{l\in \Delta}\{\om_{l}^{2}t_l\}}} \sum_{k \geq 2} \| x_{T_{k-1}} \|_{1,\bom} \\
& \leq \| A x \|_{2} + \frac{\sqrt{2(1+\delta)} 
\| G \|_{2}}{\sqrt{\min_{l\in \Delta}\{\om_{l}^{2}t_l\}}} \|
x_{T^c_0} \|_{1,\bom} \\
& \leq \| A x \|_{2} + \frac{\sqrt{1+\delta} 
\| G \|_{2}}{\min_{l\in \Delta}\{\om_{l}\sqrt{t_l/2}\}}
\| x_{\Theta^{c}} \|_{1,\bom}.
\end{align*}
This results in
\begin{align*}
\| x_{\Theta} \|_{2} 
\leq& 
\sqrt{\frac{1+\delta}{1-\delta}} \| G \|_2 \| G^{-1}
\|_{2} \frac{\sqrt{\Sboms}}{\min_{l\in \Delta}
\{\om_{l} \sqrt{t_l/2} \}}  \frac{\| x_{\Theta^c}
\|_{1,\bom}}{\sqrt{\Sboms}} + \frac{\| G^{-1} \|_2}{\sqrt{1-\delta}}
\| A x \|_{2} \\
\leq& 
 \rho \frac{\| x_{\Theta^c}
\|_{1,\bom}}{\sqrt{\Sboms}} + \sqrt{2}\| G^{-1} \|_2
\| A x \|_{2} 
\end{align*}
which establishes the weighted rNSPL of order $(\s, \M)$ with $0 < \rho < 1$ and $\gamma =
\sqrt{2} \|G^{-1}\|_{2}$.
\end{proof}

\subsection{Proof of Theorem \ref{thm:GRIPL_imply_uniform}}
\begin{proof}[Proof of Theorem \ref{thm:GRIPL_imply_uniform}]
First notice that for $0 < \delta \leq 1/2$ we have 
\[ \frac{1+\delta}{1-\delta} \leq 3.\]
Hence using Theorem \ref{thm:GRIPLimpliesrNSP} with $0< \delta_{\t,\M} \leq \delta \leq 1/2$ and 
$\rho = \sqrt{3}/2$ we see that Equation 
\eqref{eq:t_kcond}, simplifies to
Equation \eqref{eq:min_tl}. This implies that 
$AP_M$ satisfies the weighted rNSPL of order $(\s,\M)$, with constants $\rho = \sqrt{3}/2$ and $\gamma =
\sqrt{2}\|G^{-1}\|_2$. 
Now since 
\[\om_{r+1}\geq
\sqrt{\Sboms}(\tfrac{1}{3}(1+(\Sboms/\zetas)^{1/4})^{-1} 
+2\sqrt{2}\|AP_{K}^{M}\|_{1\to 2} \|G^{-1}\|_2)\] we know
from Theorem \ref{lem:min_om_bound} that any solution $\hat{x}$ of 
\eqref{eq:QCBP_542} satisfies \eqref{eq:error_bound1} and \eqref{eq:error_bound2}.
\end{proof}

\subsection{Proof of Theorem \ref{t:RIPlevels1}}\label{s:proof}

\begin{proof}[{Proof of Theorem \ref{t:RIPlevels1}}]
    We recall that $U \in \Bs(\ell^2)$ is an isometry and that 
\begin{equation*}
A = \left [ \begin{array}{c} 1/\sqrt{p_1} P_{\Omega_1} U P_{M}\\ 1/\sqrt{p_2} P_{\Omega_2} U P_{M} \\ \vdots \\ 1/\sqrt{p_r} P_{\Omega_r} U P_{M} \end{array} \right ] \in \mathbb{C}^{m\times M},
\quad 
 \text{ where }\quad p_k = m_k/(N_k - N_{k-1}), 
\end{equation*}
and $m = m_1+\ldots + m_r$.
Note that
\begin{equation*}
\| A x \|^2 - \| G x \|^2 = \ip{(A^*A - G^* G) x}{x},
\end{equation*}
and therefore
\begin{equation*}
\delta = \sup_{\Theta \in E_{\mb{s},\mb{M}}} \| P_{\Theta} (A^* A - G^* G) P_{\Theta} \|_{2}.
\end{equation*}
Notice also that $p_k = 1$ and $\Omega_{k} = \{ N_{k-1}+1,\ldots,N_k\}$ for $k=1,\ldots,r_0$.  Next notice that the matrix $P_{\Omega_k}$ can be written as
\begin{equation*}
P_{\Omega_k} = \sum^{m_k}_{i=1} e_{t_{k,i}} e^*_{t_{k,i}},
\end{equation*}
where $\{ e_i \}^{\infty}_{i=1}$ is the standard basis on $\ell^2(\mathbb{N})$.  It now follows that
\begin{align*}
\label{AstarA_sum}
A^* A &= \sum^{r}_{k=1} \frac{1}{p_k} P_{M} U^* P_{\Omega_k} U P_{M}=
\sum^{r}_{k=1} \frac{1}{p_k} \sum^{m_k}_{i=1} P_{M} U^*e_{t_{k,i}}
e^*_{t_{k,i}} U P_{M} \\
&= P_{M} U^* P_{N_{r_0}} U P_{M}+ \sum^{r}_{k=r_0+1}
\sum^{m_k}_{i=1}  X_{k,i} X^*_{k,i},
\end{align*}
where $X_{k,i}$ are random vectors given by $X_{k,i} = \frac{1}{\sqrt{p_k}} P_{M} U^* e_{t_{k,i}}$.  Note that the $X_{k,i}$ are independent, and also that
\begin{align}
\mathbb{E}(A^*A) &= P_{M} U^* P_{N_{r_0}} U P_{M}+ \sum^{r}_{k=r_0+1} \sum^{m_k}_{i=1} \mathbb{E} \left ( X_{k,i} X^*_{k,i} \right ) \nn
\\
& = P_{M} U^* P_{N_{r_0}} U P_{M} + \sum^{r}_{k=r_0+1} \frac{m_k}{p_k(N_k-N_{k-1})} \sum^{N_k}_{j=N_{k-1}+1} P_{M} U^* e_j e^*_j U P_{M}\nn
\\
& = P_{M} U^* P_{N_{r_0}} U P_{M}+ P_{M} U^* P_{N_r}^{N_{r_0}} U P_{M}\nn
\\
& = P_{M} U^* P_{N} U P_{M} 
\\
& = G^2,
\label{AstarA_expectation}
\end{align}
where $G \in \mathbb{C}^{M \times M}$ is non-singular by assumption.
Let
\begin{equation*}
D_{\mb{s},\mb{M},G} = \left \{ \eta \in \C^{M} : \| G \eta \|_2 \leq 1, \ | \Supp(\eta) \cap \{ M_{k-1}+1,\ldots,M_k \} | \leq s_k,\ k=1,\ldots,r \right \}.
\end{equation*}
We now define the following seminorm on $\mathbb{C}^{M \times M}$:
\begin{equation*}
\tnm{B}_{\mb{s},\mb{M},G} := \sup_{z \in D_{\mb{s},\mb{M},G}} \left| \ip{B z}{z} \right|,
\end{equation*}
so that
\begin{equation*}
\delta_{\mb{s},\mb{M}} = \tnm{A^* A - G^* G }_{\mb{s},\mb{M}}.
\end{equation*}
Due to \eqref{AstarA_sum} and \eqref{AstarA_expectation}, we may rewrite this as
\begin{equation}
\label{delta_norm_equiv}
\delta_{\mb{s},\mb{M}}= \tnm{\sum^{r}_{k=r_{0}+1} \sum^{m_k}_{i=1} \left (
X_{k,i} X^*_{k,i} - \mathbb{E} (X_{k,i} X^*_{k,i}) \right ) }_{\mb{s},\mb{M}}.
\end{equation}
Having detailed the setup, the remainder of the proof now follows along very similar lines to that of \cite[Thm.\ 3.2]{Li17}.  Hence we only sketch the details.

The first step is to estimate $\mathbb{E} \left ( \delta_{\mb{s},\mb{M}} \right )$.  Using the standard techniques of symmetrization, Dudley's inequality, properties of covering numbers, and arguing as in \cite[Sec.\ 4.2]{Li17}, we deduce that
\begin{equation}
\label{delta_expectation_summary}
\mathbb{E} \left ( \delta_{\mb{s},\mb{M}} \right ) \leq D + D^2,\qquad D = C_1 \sqrt{\frac{r Q \nm{G^{-1}}^2_2 \log(2 \tilde{m}) \log(2M) \log(2s)}{m}}, 
\end{equation}
where $C_1 > 0$ is a universal constant, $\tilde{m} = \sum^{r}_{k=r_0+1} m_k$, and
\begin{equation}
\label{Q_def_summary}
Q = \max_{k=r_0+1,\ldots,r} \sum^{r}_{l=1} \frac{\mu_{k,l} s_l}{p_k}.
\end{equation}
In particular, 
\begin{equation*}
\mathbb{E} \left ( \delta_{\mb{s},\mb{M}} \right ) \leq \delta/2,
\end{equation*}
provided
\begin{equation}
\label{exp_bound_summary}
C_2 Q \nm{G^{-1}}^2_2 \delta^{-2} r \log(2 \tilde{m}) \log(2M) \log^2(2s) \leq 1,
\end{equation}
where $C_2> 0$ is a constant. 
Using this, Talagrand's theorem and using the fact that $\| P_N U P_M \|_2 \leq \| U \|_2 = 1$ (see \cite[Sec.\ 4.3]{Li17}) we deduce that 
\begin{equation*}
\mathbb{P}(\delta_{\mb{s},\mb{M}} \geq \delta ) \leq \exp \left ( -3 \delta^2 /
(8 (3+7\delta) Q \nm{G^{-1}}^2_2) \right ).
\end{equation*}
In particular,
\begin{equation*}
\mathbb{P}(\delta_{\mb{s},\mb{M}} \geq \delta ) \leq \epsilon,
\end{equation*}
provided
\begin{equation*}
\frac{80}{3}Q \nm{G^{-1}}^2_2 \delta^{-2} \log(\epsilon^{-1}) \leq 1.
\end{equation*}
Combining this with \eqref{Q_def_summary} and \eqref{exp_bound_summary}  now
completes the proof.

\end{proof}

\subsection{Proof of Corollary \ref{cor:overall} and Lemma \ref{lem:bal_prop}}

\begin{proof}[Proof of Corollary \ref{cor:overall}]

    We must ensure that all the conditions are met to be able to apply Theorem \ref{thm:GRIPL_imply_uniform}
    with $P_Kx$. 

    First notice that 
    for weights $\bom = (s_{1}^{-1/2}, \ldots, s_{r}^{-1/2}, \om_{r+1})$
    we have $\Sboms = r$ and $\zetas = 1$. Next we note that condition $(ii)$
    implies that $P_K x$ is a feasible point since $\|HP_K x - \tilde{y}\|_{2} 
    \leq \|HP_{K}^{\perp} x\|_{2} + \|e_1\|_{2} = \eta + \eta'$.

    Let $G = \sqrt{P_{M}U^*P_N UP_M}$. Combining condition $(i)$ and Lemma \ref{lem:bal_prop} gives 
    $\|G^{-1}\|_2 \leq 1/\sqrt{\theta}$ and since $\|G\|_{2} \leq 1$ we also have 
    $\kappa(G) = \|G\|_{2}\|G^{-1}\|_{2} \leq 1/\sqrt{\theta}$. Inserting the above 
    equalities and inequalities into the weight condition for $\om_{r+1}$ 
    in Theorem \ref{thm:GRIPL_imply_uniform} gives condition $(iii)$. 

    Next we must ensure that $AP_M$ satisfies the G-RIPL of order 
    $(\t,\M)$ with $\delta_{\t,\M} \leq 1/2$ where 
    \begin{equation}
        t_l = \min\left\{ M_{l}-M_{l-1}, 2\ceil{4\theta^{-1}rs_l} \right\}.
    \end{equation}
    According to Theorem \ref{t:RIPlevels1} this occurs if the $m_k$'s satisfies 
    condition $(iv)$. The error bounds \eqref{eq:error_bound1} and \eqref{eq:error_bound2}
    now follows directly from Theorem \ref{thm:GRIPL_imply_uniform}. 
\end{proof}

\begin{proof}[Proof of lemma \ref{lem:bal_prop}]
First notice that the balancing property is equivalent to requiring 
\begin{equation}
\sigma_M(P_N U P_M) \geq \sqrt{\theta}    
\label{eq:bl_prop_eq}
\end{equation}
where $\sigma_{M}(P_{N} UP_M)$ is the $M$th largest singular value of $P_NU
P_M$.  Indeed, since $U$ is an isometry, the matrix $P_M - P_M U^* P_N UP_M$ is
nonnegative definite, and therefore 
\begin{align}
    \|P_MU^*P_NU_PM -P_M\|_{2} 
&= \sup_{x \in \C^M, \|x\|_{2}\leq 1} \ind{(P_MU^*P_NU_PM -P_M)x,x} \\
&= \sup_{x \in \C^M, \|x\|_{2}\leq 1} \( \|P_Mx\|_2 - \|P_NUP_Mx\|_{2} \) \\
&= 1 - \inf_{x \in \C^M, \|x\|_{2}= 1} \|P_NUP_Mx\|_2
\label{eq:bal_prop2}
\end{align}
This gives \eqref{eq:bl_prop_eq}. Next let $G = \sqrt{P_M U^*P_N UP_M}$ and
notice that $\sigma_M(G) = \sigma_{M}(P_N UP_M)$. 
This gives $\|G^{-1}\|_{2} =
1/\sigma_{M}(G) \leq 1/\sqrt{\theta}$. 
\end{proof}

\section{Proof of results in Section \ref{sec:4}}
\label{sec:coher}
In Section \ref{sec:4} we found concrete recovery guarantees for the Walsh
sampling and wavelet reconstruction, using the theorems in Section
\ref{sec:infin_dim}. The key to deriving Walsh-wavelet recovery guarantees
boils down to estimating the quantities $\mu_{k,l}$, $||HP_{K}^{M}||_{1 \to 2}$ and
$||G^{-1}||_{2} \leq \frac{1}{\sqrt{\theta}}$. All of these quantities depend 
directly $U = [\Bwh, \Bwave^{J_0,\nu}]$, and to control them  we will have to 
estimate how the entries of $U$ changes for varying $n, j,k$ and $s$. We will 
therefore start this section by setting up notation for
wavelets on the interval and stating some useful properties of Walsh functions.
Then in Section \ref{subsec:mukl} and \ref{subsec:loc_coherence} we will estimate $\mu_{k,l}$,
followed by a discussion of the sharpness of this estimate for $\nu=2$ in
Section \ref{subsec:sharpness_coher}. We will then finish in Section \ref{sec:proof2}
by estimating $||HP_{K}^{M}||_{1\to 2}$, show how $\theta$ scales for varying
$M$ and $N$, and prove Theorem \ref{thm:1d_GRIPL} and \ref{thm:1d_GRIPL_rec}.

\subsection{Wavelets on the interval and regularity}
In section \ref{subsec:wavelets1} we introduced orthogonal wavelets on the 
real line, but we did not make any formal definitions of the wavelets we 
used at the boundaries of the interval $[0,1)$. Next we consider the two
boundary extensions, \emph{periodic} and \emph{boundary wavelets}. To simplify the 
exposition we define the following sets
\begin{align*}
    &\Lanul \coloneqq \{0, \ldots, \nu-1\}, 
    &&\Lanum \coloneqq \{\nu, \ldots, 2^{j}-\nu-1\}, \\
    &\Lanur \coloneqq \{ 2^{j}-\nu, \ldots, 2^{j} - 1 \}
    && \La_{j} = \Lanul \cup \Lanum \cup \Lanur
\end{align*}

At each scale $j \geq J_0$, the periodic wavelet basis consists of the usual
wavelets and scaling functions  $\psi_{j,k}$, $\phi_{j,k}$ for $k \in \Lanum$
and the periodic extended functions $\phi_{j,k}^{\text{per}}$ and
$\psi_{j,k}^{\text{per}}$ for $k \in \Lanul \cup \Lanur $. These are defined as
 \begin{align}
     \phi_{j,k}^{\text{per}} &\coloneqq \phi_{j,k}\lvert_{[0,1)} + \phi_{j,2^j + k}\lvert_{[0,1)} 
     && \text{for } k \in \Lanul \\
     \phi_{j,k}^{\text{per}} &\coloneqq \phi_{j,2^j-\nu-k}\lvert_{[0,1)} + \phi_{j,k}\lvert_{[0,1)} 
     && \text{for } k \in \Lanur 
\end{align}
and similarly for $\psi_{j,k}^{\text{per}}$. Strictly speaking we could have
defined these periodic extensions only for $k = 0, \ldots, \nu-2$ and $k =
2^{j}-\nu+1, \ldots, 2^j-1$, but to unify the notation for both boundary
extensions we have chosen the former. 

Next we have the boundary wavelet basis with $\nu$ vanishing
moments. This wavelet basis consists of the same interior wavelets as the
periodic basis, but with $2\nu$ boundary scaling and wavelet functions.
\[
    \phi_{k}^{\textnormal{left}},  \phi_{k}^{\textnormal{right}},
    \psi_{k}^{\textnormal{left}},  \psi_{k}^{\textnormal{right}}, 
    ~~~ \textnormal{ for } k = 0,\ldots,\nu-1.
\]
As for the interior functions we also define the scaled versions as 
\begin{equation}
    \label{eq:db_bw_scale}
    \begin{array}{ll}
\phi_{j,k}^{\textnormal{left}}(x) \coloneqq 2^{j/2}\phi_{k}^{\text{left}} (2x), 
&\phi_{j,k}^{\textnormal{right}}(x) \coloneqq 2^{j/2}\phi_{k}^{\text{right}}
    (2x), \\
\psi_{j,k}^{\textnormal{left}}(x) \coloneqq 2^{j/2}\psi_{k}^{\text{left}} (2x), 
&\psi_{j,k}^{\textnormal{right}}(x) \coloneqq 2^{j/2} \psi_{k}^{\text{right}}
        (2x).\\ 
    \end{array}
\end{equation}
The names 'left' and 'right' corresponds to the support of these functions. That is 
\begin{align*}
    \Supp \phi_{j,k}^{\text{left}}  = [0, 2^{-j}(\nu+k)] \\ 
    \Supp \phi_{j,k}^{\text{right}} = [2^{-j}(2^j -\nu - k),1]  
\end{align*}
for $k = 0, \ldots, \nu-1$.

In the following we shall see that all of our results holds for both periodic
and boundary wavelets, but their treatment in some of the proofs differs slightly.
To make the treatment as unified as possible we make the following definition.
\begin{definition}
    We say that $\phi_{j,k}^{s}$, $s \in \{0,1\}$ \textquote{originates from a periodic
    wavelet} if 
    \begin{align*}
        \phi_{j,k}^{0} &\coloneqq
        \begin{cases} 
            \phi_{j,k}^{\text{per}} & \text{for } k \in \Lanul \\
            \phi_{j,k} & \text{for } k \in \Lanum  \\
            \phi_{j,k}^{\text{per}} & \text{for } k \in \Lanur  \\
        \end{cases}, & 
        \phi_{j,k}^{1} &\coloneqq 
        \begin{cases} 
            \psi_{j,k}^{\text{per}} & \text{for } k \in \Lanul\\
            \psi_{j,k} & \text{for } k \in \Lanum \\
            \psi_{j,k}^{\text{per}} & \text{for } k \in \Lanur \\
        \end{cases}. 
    \end{align*}
    We say that $\phi_{j,k}^{s}$ \textquote{originates from a boundary wavelet} if 
    \begin{align*}
        \phi_{j,k}^{0} &\coloneqq
        \begin{cases} 
            \phi_{j,k}^{\text{left}} & \text{for } k \in \Lanul \\
            \phi_{j,k} & \text{for } k \in \Lanum  \\
            \phi_{j,2^j -1 - k}^{\text{right}} & \text{for } k \in \Lanur  \\
        \end{cases}, & 
        \phi_{j,k}^{1} &\coloneqq 
        \begin{cases} 
            \psi_{j,k}^{\text{left}} & \text{for } k \in \Lanul\\
            \psi_{j,k} & \text{for } k \in \Lanum \\
            \psi_{j,2^j -1 - k}^{\text{right}} & \text{for } k \in \Lanur \\
        \end{cases}. 
    \end{align*}
\end{definition}
With these functions defined now for both boundary extensions, the definition
of $\Bwave^{J_0, \nu}$ is also clear. Next we make a note on the
regularity of these orthogonal wavelets.

\begin{definition}
   Let $\alpha = k + \beta$, where $k \in \Z_{+}$ and $0< \beta < 1$. A function 
$f\colon \R \to \R$ is said to be uniformly Lipschitz $\alpha$ if $f$ is $k$-times continuously differentiable and for which
the $k^{\text{th}}$ derivative $f^{(k)}$ is H\"older continuous with exponent $\beta$, i.e.
\[ |f^{(k)}(x) - f^{(k)}(y)| < C |x-y|^{\beta}, \quad \forall x,y \in \R \]
for some constant $C > 0$.  
\end{definition}

In particular the Daubechies wavelet with 1 vanishing moment (i.e., the Haar
wavelet) is not uniformly Lipschitz as it is not continuous, whereas for $\nu
\geq 2$ we have the constants found in table \ref{table:lipschitz_reg} \cite[239]{Daubechies92}.
For large $\nu$, $\alpha$ grows as $0.2\nu$ \cite[294]{Mallat09}. Also note
that each of the boundary functions $\phi_{k}^{\text{left}},
\phi_{k}^{\text{right}}$ and $\psi^{\text{left}}_{k},\phi_{k}^{\text{right}}$
are constructed as finite linear combinations of the interior scaling function
$\phi$ and wavelet $\psi$. Thus all of these boundary functions has the same
regularity as $\phi$ and $\psi$. 
\begin{table}[htb]
    \begin{center}
    \begin{tabular}{cc}
        $\nu$ & $\alpha$ \\ \hline 
        2 & 0.55 \\ 
        3 & 1.08 \\
        4 & 1.61 
    \end{tabular}
    \end{center}
    \vspace*{-\baselineskip}
    \caption{The Lipschitz regularity of Daubechies wavelets with $\nu$ vanishing 
    moments.}
    \label{table:lipschitz_reg}
\end{table}

\subsection{Properties of Walsh functions}
\begin{definition}
    Let $x = \{x_i\}_{i=1}^{\infty}$ and $y = \{y_i\}_{i=1}^{\infty}$ be 
    sequences consisting of only binary numbers. That is $x_i, y_i \in \{0,1\}$
    for all $i \in \N$. The operation $\oplus$ applied to these sequences gives
    \begin{equation}
        x \oplus y \coloneqq \{ |x_i - y_i|\}_{i=1}^{\infty}. 
    \end{equation}
    For two binary numbers $x_i, y_i \in \{0,1\}$, we let $x_i \oplus y_i = 
    |x_i-y_i|$. 
\end{definition}

\begin{proposition}
    For $j, m,n \in \Z_+$ and $x,y \in [0,1)$, the Walsh function satisfies the 
    the following properties 
    \begin{align}
        \int_{0}^{1} w_n(x)w_m(x) \d x &= \begin{cases} 
                                            1 & \textnormal{if } m = n \\
                                            0 & \textnormal{otherwise}
                                         \end{cases} 
                                         \label{eq:ortho}\\
        w_{n}(x \oplus y) &= w_n(x)w_n(y) \label{eq:oplus1} \\ 
        w_n(2^{-j}x) &= w_{\floor{n/2^j}}(x) \label{eq:floor}
    \end{align}
\end{proposition}
\begin{proof}
    Equation \eqref{eq:oplus1} and \eqref{eq:ortho} can be
    found in any standard text on Walsh functions e.g., \cite{Golubov91},
    whereas the last follows by inserting $j$ zeros in front of $x$'s
    dyadic expansion. 
\end{proof}


\subsection{Bounding the inner product $|\indi{\phi_{j,k}^{s},w_n}|$}
\label{subsec:mukl}
The entries in $U = [\Bwh, \Bwave^{J_0, \nu}]$, consists of $\indi{\phi_{j,k}^{s},w_n}$
for different values of $j,k,s$ and $n$. Thus in order to determine the local 
coherences we need to find an upper bound of this inner product. Next we 
derive such an bound for $\nu \geq 2$ vanishing moments and discusses its
sharpness.  For $\nu=1$ we determine the magnitude of each matrix 
entry explicitly.

\begin{lemma}
    \label{lem:ind_wave_had_1}
     Let $w_n \in \Bwh$ and let $\phi_{j,k}^{s} \in \Bwave^{J_0, \nu}$ for $\nu
    \geq 2$. For $j \geq J_{0}$, $s\in \{0,1\}$ and $k \in \La_j$ we have
        \begin{align}
            \left|\ind{\phi_{j,k}^{s}, w_n} \right| \leq 2^{-j/2} 2\nu \max_{l
            \in \Gamma_k}
            \left\{ \left|\Ws \phi^s(\cdot + l)\big\lvert_{[0,1)}
            \( \floor{\frac{n}{2^{j}}}\) \right| \right\}
            \label{eq:firstInnerprod}
        \end{align}
        where
        \begin{align*}
            \Gamma_k = \begin{cases} 
             \{0, \ldots, \nu+k-1\}& \textnormal{for }
             k \in \Lanul;  \\
             \{-\nu+1, \ldots, \nu - 1\}& \textnormal{for } 
             k \in \Lanum; \\
             \{k-\nu+1, \ldots, 2^j - 1\}& \textnormal{for }   
                      k \in \Lanur. 
            \end{cases}
        \end{align*}
        if $\phi_{j,k}^{s}$ originates from a boundary wavelet and 
        \[ \Gamma_k = \{ -\nu+1, \ldots, \nu-1\}\]
        if $\phi_{j,k}^{s}$ originates from a periodic wavelet.
\end{lemma}
\begin{proof}
First notice that for any $x \in [0,1)$ we have
    \begin{align*}
        \frac{x}{2^{j}} + \frac{k}{2^{j}} 
        &= \sum_{i = j}^{\infty} x_{i-j+1} 2^{-i-1}
        + \sum_{i=1}^{j} k_i 2^{-j-1+i}\\
        &= \sum_{i = j}^{\infty} x_{i-j+1} 2^{-i-1}
          \oplus \sum_{i=1}^{j} k_i 2^{-j-1+i}
        = \frac{x}{2^{j}} \oplus \frac{k}{2^{j}} .
    \end{align*}
Next, we only consider the interior wavelets $\phi_{j,k}^{s}$ i.e. $k\in
\Lanum$. For $k \in \Lanul \cup \Lanur$, we need to handle the two cases 
where $\phi_{j,k}^{s}$ orignates from a periodic and boundary wavelet
seperately. The arguments/calculations for the two different boundary extensions are analogous. Also, both of these extensions will have support less than $2\nu$. 
    
For $k \in \Lanum$, notice that $\supp (\phi_{j,k}^{s}) =
[2^{-j}(-\nu+1+k), 2^{-j}(\nu+k)]$. 
    
    { \allowdisplaybreaks
    \begin{align*}
        \ind{\phi^{s}_{j,k}, w_n} &= \int_{0}^{1} \phi_{j,k}^{s}(x)w_n(x)\d x\\ 
        &= \int_{2^{-j}(-\nu+1+k)}^{2^{-j}(\nu+k)} 2^{j/2} \phi^{s}(2^jx-k) w_n(x) \d x \\
        &=  2^{-j/2} \int_{-\nu+1}^{\nu} 
                  \phi^{s}\(x\) w_n\(\frac{x+k}{2^{j}}\) \d x \\
        &= 2^{-j/2} \sum_{l=-\nu+1}^{\nu-1} \int_{0}^{1} 
                  \phi^{s}\(x+l\) w_n\(\frac{x+l+k}{2^{j}}\) \d x \\
        &= 2^{-j/2} \sum_{l=-\nu+1}^{\nu-1} \int_{0}^{1} 
                  \phi^{s}\(x+l\) w_n\(\frac{x}{2^{j}} \oplus\frac{l+k}{2^{j}}\) \d x \\
        &= 2^{-j/2} \sum_{l=-\nu+1}^{\nu-1} w_n \(\frac{l+k}{2^{j}}\)  \int_{0}^{1} 
                  \phi^{s}\(x+l\) w_n\(\frac{x}{2^{j}}\) \d x \\
        &= 2^{-j/2} \sum_{l=-\nu+1}^{\nu-1} w_n \( \frac{l+k}{2^{j}}\)
          \Ws \phi^{s}_{0,-l}\big\lvert_{[0,1)} \(\floor{\frac{n}{2^{j}}}\)
            \label{eq:sum_wal_phi} \\
        &\leq 2^{-j/2} 2\nu \max_{l
            \in \Gamma_k}
            \left\{ \left|\Ws \phi^s(\cdot + l)\big\lvert_{[0,1)}
            \( \floor{\frac{n}{2^{j}}}\) \right| \right\}
    \end{align*}
    }
\end{proof}

\begin{lemma}[\cite{Butzer75}]
    \label{lem:regularity_bound}
    Let $f:[0,1) \to \R$ be uniformly Lipschitz $0 < \alpha \leq 1$ then 
    \[ |\Ws f(n) | = \left| \int_{0}^{1} f(x) w_n(x) \d x \right| \lesssim (n+1)^{-\alpha}  \]
    for $n\in
    \Z_+$. 
\end{lemma}

\begin{theorem}
    \label{thm:bound_ip}    
Let $\phi_{l,t}^{s} \in \Bwave^{J_0, \nu}$ with $\nu\geq 3$ and let $w_n \in
\Bwh$. 
For $l\geq J_0$ and $2^{k}\leq n < 2^{k+1}$ with $k\in \Z_+$, we have 
    \[ |\ind{\phi_{l,t}^{s}, w_n}|^2 \lesssim 2^{-k}2^{-|l-k|}\]
    for all $t \in \La_l$ and $s\in \{0,1\}$. For $n=0$ the bound hold with
    $k=0$. 
\end{theorem}
\begin{proof}
    To obtain the bound above we will combine Lemma \ref{lem:ind_wave_had_1}
and Lemma \ref{lem:regularity_bound}. We start by arguing that $\phi_{l,t}^{s}$
have the same regularity regardless of boundary extension. 
Let $a \in \Gamma_t$ where 
$\Gamma_t$ is as in lemma \ref{lem:ind_wave_had_1}. 

If $\phi_{l,t}^{s}$ originates from a
periodic wavelet, $\phi^{s}_{0,-a}\lvert_{[0,1)}$,  will have Lipschitz regularity
$\al > 0$,  since both $\phi$ and $\psi$ have this regularity. Next if
$\phi_{l,t}^{s}$ originates from a boundary wavelet and $t\in \Lambda_{\nu,l,\text{mid}}$,
$\phi_{0,-a}^{s}\lvert_{[0,1)}$ will have Lipschitz regularity $\al$, by the
same argument as above.  If $t \in \Lambda_{\nu,l,\text{left}} \cup \Lambda_{\nu, l,\text{right}}$ we know from the
construction of the boundary functions \cite{Cohen93} that these are finite
linear combinations of $\phi_{l,t}$ and $\psi_{l,t}$. These function will
therefore posses the same regularity $\al$ as the interior function.
    
    Next notice from table \ref{table:lipschitz_reg} that for $\nu \geq 3$
    vanishing moments, we known that $\alpha\geq 1$. Applying Lemma \ref{lem:ind_wave_had_1}
    and Lemma \ref{lem:regularity_bound} then gives
        \begin{align}
            \left|\ind{\phi_{l,t}^{s}, w_n} \right|^2 &\leq 2^{-l} 4\nu^2 \max_{a
            \in \Gamma_t}
            \left\{ \left|\Ws \phi^s(\cdot + a)\big\lvert_{[0,1)}
            \( \floor{\frac{n}{2^{l}}}\) \right|^2 \right\} \\
            &\lesssim 2^{-l} \frac{1}{(\floor{\frac{n}{2^{l}}}+1)^2}
            \leq 2^{-l} \frac{1}{(\floor{2^{k-l}}+1)^2} \leq 2^{-k} 2^{-|l-k|}
        \end{align}
        where $\Gamma_t$ depends on the boundary extension.
\end{proof}

\begin{theorem}
    \label{thm:bound_ip_haar}
    Let $w_n \in \Bwh$
    and let $\phi_{l,t}^{s} \in \Bwave^{J_0, 1}$ for $l\geq 0$ and $t\in \La_l$. Then 
    \begin{align*}
        |\ind{\phi_{l,t}^{0}, w_n}|^2 &= 
    \begin{cases} 2^{-l} &\text{if } n < 2^{l}\\
                  0 &\text{otherwise} 
    \end{cases} \\
        |\ind{\phi_{l,t}^{1}, w_n}|^2 &= 
    \begin{cases} 2^{-l} &\text{if } 2^{l}\leq n < 2^{l+1}\\
                  0 &\text{otherwise} 
    \end{cases}.
    \end{align*}
\end{theorem}
\begin{proof}
    These equalities can be found in either \cite{Antun16} or \cite{Terhaar18}.
\end{proof}

\subsection{Proof of Proposition \ref{lem:1d_coher_bound},
\ref{lem:Haar_isometry} and \ref{lem:1d_coher_bound_haar}}
\label{subsec:loc_coherence}
Using the above results we are now able to determine the local coherences of 
$U = [\Bwh, \Bwave^{J_0, \nu}]$. 

\begin{proof}[Proof of Proposition \ref{lem:1d_coher_bound}]
    We use the bound found in Theorem \ref{thm:bound_ip}. 
    Recall that $ \M = [2^{J_0+1}, \ldots, 2^{J_0+r}]$ and 
$\Nm = [2^{J_0+1}, \ldots, 2^{J_0-1+r}, 2^{J_0+r+q}]$. For fixed $l \in \{ 1,\ldots,r\}$ and $k \in \{2,\ldots, r\}$ we have 
    \begin{align*}
        \mu_{k,l} &= \max\left\{ |\ind{\phi_{J_0-1+l,t}^{s}, w_n}|^2 : 
\begin{subarray}{c}  N_{k-1} \leq n < N_{k} \\ 
t \in \La_{J_0-1+l}, s\in \{0,1\}, \text{ if } l=1,\\s = 1 \text{ if } l>1   
\end{subarray}    
\right\}
\\
&\lesssim
2^{-(J_0-1+k)} 2^{-|(J_0-1+l)-(J_0-1+k)|}
\lesssim 
2^{-J_0 -k} 2^{-|l-k|}
\end{align*}
For $l \in \{1,\ldots, r\}$ and $k = 1$ we have $N_0=0$. This gives 
\begin{align*}
        \mu_{k,l} &= \max\left\{ |\ind{\phi_{J_0-1+l,t}^{s}, w_n}|^2 : 
\begin{subarray}{c}  0 \leq n < N_{1} \\ 
t \in \La_{J_0-1+l}, s\in \{0,1\}, \text{ if } l=1,\\s = 1 \text{ if } l>1   
\end{subarray}    
\right\}
\\
&\lesssim
2^{-(J_0-1+l)}
\lesssim 
2^{-J_0 -k} 2^{-|l-k|}
\end{align*}
\end{proof}

\begin{proof}[Proof of Proposition \ref{lem:Haar_isometry}]
    Since both $\Bwave^{J_0, 1}$ and $\Bwh$ are orthonormal, $U = [\Bwh, \Bwave^{J_0, 1}]$
is an isometry on $\ell^2(\N)$ i.e. $U^*U = I \in \Bs(\ell^{2}(\N))$. Let
$N = 2^k$ for some $k \in \N$ with $k \geq J_0+1$. Using Theorem
\ref{thm:bound_ip_haar} we see that 
\begin{align*}
    P_{N}^{\perp}UP_{N} 
&= \left\{
\ind{\phi_{j, t}^{s}, w_n} : 
\begin{subarray}{c} n \geq 2^k \\
s=1, J_0 \leq j < k, t \in \La_j \\
s=0, j=J_0, t \in \La_{J_0}
\end{subarray}
\right\} = 0
\end{align*}
which means that 
\begin{align*}
    (P_NUP_N)^* (P_NUP_N) &= 
((P_N+P_{N}^{\perp}) U P_{N})^* ( (P_N+P_{N}^{\perp})U P_N) \\
&= (UP_N)^* (UP_N) = P_NU^*UP_N = I \in \C^{N\times N}.
\end{align*}
\end{proof}

\begin{proof}[Proof of Proposition \ref{lem:1d_coher_bound_haar}]
    We use the bound found in Theorem \ref{thm:bound_ip_haar}. 
    Recall that $ \M = \Nm = [2^{J_0+1}, \ldots, 2^{J_0+r}] $. For fixed $k,l \in \{ 1,\ldots,r\}$ we have that 
    \begin{align*}
        \mu_{k,l} &= \max\left\{ |\ind{\phi_{J_0-1+l,t}^{s}, w_n}|^2 : 
\begin{subarray}{c}  N_{k-1} \leq n < N_{k} \\ 
t \in \La_{J_0-1+l}, s\in \{0,1\}, \text{ if } l=1,\\s = 1 \text{ if } l>1   
\end{subarray}    
\right\}
\\
&= 
\begin{cases} 
2^{-J_0-l+1} & \text{if } l=k \\
0 & \text{otherwise}
\end{cases}.
\end{align*}
\end{proof}

\subsection{About the sharpness of the local coherence bounds}
\label{subsec:sharpness_coher}
As can be seen from Proposition \ref{lem:1d_coher_bound_haar}, the coherence bounds 
for $\nu=1$ are sharp. However, for $\nu \geq 2$, we have not discussed their 
sharpness. In fact, none of the results in this paper consider the case for $\nu=2$
vanishing moments. The reason for this is that these wavelet have a 
Lipschitz regularity $\alpha \approx  0.55$, which means that the bound in 
Theorem \ref{thm:bound_ip} would have less rapid decay if we had included these wavelets in the theorem. To simplify the presentation we have chosen to exclude them.  

We will argue that Theorem \ref{thm:bound_ip} does not seem to extend to
wavelets with $\nu=2$ vanishing moments. Let $\M = \Nm = [2^{J_0+1}, \ldots,
2^{J_0+r}]$ and $U = [\Bwh, \Bwave^{J_0, \nu}]$ for $\nu\geq 2$.  Notice that setting 
$\nu = 2$
does only affect the local coherence estimates $\mu_{k,l}$ for $k \geq l$. For
$k < l$, the local coherences are unaffected by the regularity of the wavelet.
This follows from Lemma \ref{lem:ind_wave_had_1}, by setting $|\Ws \phi^s(\cdot
+l)(0)| \approx 1$.  Next consider the case where $k \geq l$, then Theorem
\ref{thm:bound_ip} suggests that $\mu_{k,l}/\mu_{k+1} \approx 4$ for $\nu\geq
3$.  

We now consider table \ref{table:mu_kl} and notice that for $\nu =
2$, all of the 18 entries in table \ref{table:mu_kl} have values less than $4$. 
This suggest that the bound in Theorem \ref{thm:bound_ip} does not extend to
the case of $\nu=2$ vanishing moments. From the same table we also observe that
for $\nu = 4$, the bound in Theorem \ref{thm:bound_ip} seem to be quite sharp.
While there are a few entries that are less than $4$, most are very close,
if not larger than this value. 
\begin{table}[htb]
    \begin{center}
\begin{tabular}{cccc}
$\mu_{k,l}/\mu_{k+1,l}$ & $l=1$ & $l=2$ & $l=3$ \\
$k=2$  & 3.017 &        &         \\ 
$k=3$  & 2.532 & 1.854 &         \\
$k=4$  & 3.292 & 2.532 &  1.846 \\ 
$k=5$  & 3.653 & 3.293 &  2.534 \\ 
$k=6$  & 3.828 & 3.653 &  3.293 \\ 
$k=7$  & 3.914 & 3.828 &  3.654 \\ 
$k=8$  & 3.957 & 3.914 &  3.828  
\end{tabular}
\quad
\begin{tabular}{cccc}
    $\mu_{k,l}/\mu_{k+1,l}$ & $l=1$ & $l=2$ & $l=3$ \\
    $k=2$ & 4.342 &  & \\
    $k=3$ & 6.160 & 3.439 &   \\ 
    $k=4$ & 3.643 & 6.202 & 3.503 \\ 
    $k=5$ & 4.060 & 3.639 & 6.286 \\ 
    $k=6$ & 3.961 & 4.064 & 3.632 \\ 
    $k=7$ & 4.004 & 3.960 & 4.070 \\ 
    $k=8$ & 3.996 & 4.004 & 3.960  
\end{tabular}

    \end{center}
    \vspace*{-\baselineskip}
    \caption{
Left: Fraction between the local coherences for $U = [\Bwh, \Bwave^{3, 2}]$ and
$\M = \Nm = [2^4, \ldots, 2^{11}]$. Right: Fraction between the local coherences
for $U = [\Bwh, \Bwave^{4, 4}]$ and $\M = \Nm = [2^5, \ldots, 2^{12}]$. }
    \label{table:mu_kl}
\end{table}

\subsection{Proof of remaining results in Section \ref{sec:4}}
\label{sec:proof2}

\begin{proof}[Proof of Proposition \ref{lem:Terhaar_1d}]
 This proposition is a consequence of Theorem 1.1 in
\cite{Terhaar17}.  
Let $\mathcal{S}_{N} = \{w_n: n = 0,\ldots, N-1\}$ and $\mathcal{R}_{M}$ be the 
$M$ first function in $\Bwave^{J_0, \nu}$. The subspace cosine angle between $\mathcal{S}_N$
and $\mathcal{R}_{M}$ is defined as 
\[
\cos(\om(\mathcal{R}_M, \mathcal{S}_{N})) = \inf_{f \in \mathcal{R}_{M}, \|f\|= 1} \|P_{\mathcal{S}_{N}} f\| \quad \text{where} \quad \om(\mathcal{R}_M, \mathcal{S}_N) \in [0,\pi/2],
\]
and $P_{\mathcal{S}_N}$ is the projection operator onto $\mathcal{S}_N$. As both $\Bwh$ and $\Bwave^{J_0,\nu}$
are orthonormal bases, the synthesis and analysis operators are unitary. We therefore have 
\begin{equation}
\inf_{f \in \mathcal{R}_{M}, \|f\|= 1} \|P_{\mathcal{S}_{N}} f\| = \inf_{x \in
\C^M, \|x\|_{2}=1} \|P_N U P_M x\|_{2}
\label{eq:cos_ang1}
\end{equation}
Furthermore notice that by equation \eqref{eq:bal_prop2} and the
definition of the balancing property, we have 
\begin{equation}
 \cos(\om(\mathcal{R}_M, \mathcal{S}_{N})) = \inf_{x \in \C^M, \|x\|_{2}=1} \|P_N U P_M x\|_{2} \geq \theta. 
\end{equation}
Hence if $U$ satisfies the balancing property of order $\theta \in (0,1)$ for $N$
and $M$, then $1/\cos(\om(\mathcal{R}_M, \mathcal{S}_{N})) \leq 1/\theta$, where 
$1/\theta > 1$.
Next for $M \in \N$ and $\gamma>1$ we define the \emph{stable sampling rate} as
\[
\Theta(M, \gamma) = \min(N\in \N: 1/\cos(\om(\mathcal{R}_M, \mathcal{S}_N)) <
\gamma).
\]
Rearranging the terms we see that if $N$, $M$ satisfies the stable sampling
rate of order $\gamma = 1/\theta > 1$ then $U$ satisfies the balancing property
of order $\theta$ for $N$ and $M$.

Theorem 1.1 in \cite{Terhaar17} states that for $M =2^r$, $r \in \N$ and for 
all $\gamma>1$ there exists a constant $S_{\gamma}>1$ (dependent on $\gamma$),
such that whenever $N \geq S_{\gamma} M$, then $1/\cos(\om(\mathcal{R}_M,
\mathcal{S}_N)) < \gamma$. Moreover, we have the relation 
$\Theta(M,\gamma) \leq S_{\gamma} M = \mathcal{O}(M)$. 
Hence if $q = \ceil{\log_{2}S_{1/\theta}}$ we see that the proposition hold with 
$N = 2^{k + q} \geq S_{1/\theta} 2^{k} > 2^{k} = M$. 
\end{proof}

\begin{proof}[Proof of Proposition \ref{prop:norm12}]
Using Theorem \ref{thm:bound_ip}, we see that $\mu(P_N U P_{K}^{\perp})
\lesssim K^{-1}$.
This gives 
\begin{align*}
    \|HP_{K}^{\perp}x\|_{2}^{2} &= \sum_{k=1}^{r} \frac{N_{k}-N_{k-1}}{m_{k}}
    \sum_{i \in \Om_{k}} \bigg| \sum_{j>K} U_{ij} x_j \bigg|^2  \\
    &\leq  \sum_{k=1}^{r} \frac{N_{k}-N_{k-1}}{m_{k}}
    \sum_{i \in \Om_{k}} \bigg ( \sum_{j>K} \sqrt{\mu(P_{N}UP_{K}^{\perp})} |x_j| \bigg )^2  \\
    &\leq  \sum_{k=1}^{r} (N_{k}-N_{k-1}) \mu(P_{N}UP_{K}^{\perp}) 
    \bigg ( \sum_{j>K} |x_j| \bigg)^2  \\
    &\leq N \mu(P_{N}UP_{K}^{\perp})
    \bigg( \sum_{j>K}  |x_j| \bigg)^2  \lesssim \frac{N}{K} \|x\|_{1}^{2}.
\end{align*}

\end{proof}

\begin{proof}[Proof of Theorem \ref{thm:1d_GRIPL}]
First recall that $\M = [2^{J_0+1},\ldots, 2^{J_0+r}]$ and
  $\Nm = [2^{J_0+1}, \allowbreak
\ldots, 2^{J_0+r-1}, 2^{J_0+r+q}]$ where $q\geq 0$ is
chosen so that $G$ satisfies the balancing property of order $0<\theta<1$. From
Lemma \ref{lem:bal_prop} we therefore have $\|G^{-1}\|_{2} \leq
1/\sqrt{\theta}$.  

From Theorem
\ref{t:RIPlevels1} we know that the matrix $A$ in equation
\eqref{def:A_infinite} satisfies the G-RIPL with $\delta_{\s,\M}
\leq \delta$, provided the sample densities $\m \in \N^r$ satisfies 
$m_{k} = N_{k}-N_{k-1}$ for $k=1,\ldots, r_0$, and 
\begin{equation}
    \label{eq:samp_cond_431}
    m_k \gtrsim  \delta^{-2} \cdot \|G^{-1}\|_{2}^{2} \cdot (N_k - N_{k-1}) \cdot \bigg (
    \sum^{r}_{l=1} \mu_{k,l} \cdot s_l \bigg ) \cdot L, 
\end{equation}
for $k=r_0+1,\ldots, r$. Next notice that $N_{k} - N_{k-1} = 2^{J_0 +k-1}$ for 
$k=2,\ldots, r-1$, while $N_{r} - N_{r-1} = 2^{J_0 +r}(2^{q}-2^{-1})$ and
$N_1-N_0 = 2^{J_0+1}$. Using the local coherences $\mu_{k,l}$ from Proposition 
\ref{lem:1d_coher_bound} we obtain 
\begin{align*}
    (N_{k}-N_{k-1})\bigg( \sum_{l=1}^{r} \mu_{k,l}s_l\bigg) 
&\lesssim 2^{J_0+k} 2^{q \max\{k+1-r, 0\}}
\bigg( \sum_{l=1}^{r} 2^{-J_0-k}2^{-|l-k|}s_l \bigg) \\
&= 2^{q \max\{k+1-r, 0\}}
\bigg( \sum_{l=1}^{r} 2^{ -|k- l|} s_l \bigg). \\
\end{align*}
Inserting this and $\|G^{-1}\|_{2}^{2} \leq \theta^{-1}$ into
\eqref{eq:samp_cond_431} leads to the sampling condition in Theorem
\ref{thm:1d_GRIPL}. 
\end{proof}

\begin{proof}[Proof of Theorem \ref{thm:1d_GRIPL_rec}]
    The theorem is identical to Corollary \ref{cor:overall}, except that we have 
    fixed $\M$ and $\Nm$. The concrete values for these have been inserted
    in condition $(iv)$ together with the local coherences $\mu_{k,l}$.
    The computation of this can be found in the proof above.
\end{proof}


\section*{Acknowledgements}
The authors would like to thank Simone Brugiapaglia, Simon Foucart, Remi Gribonval, Øyvind Ryan and Laura Thesing for useful discussions and comments.
BA acknowledges support from the Natural Sciences and Engineering Research Council of Canada through grant 611675. ACH acknowledges support from the UK Engineering and Physical Sciences Research Council (EPSRC) grant EP/L003457/1, a Royal Society University Research Fellowship, and the Philip Leverhulme Prize (2017).

    \newpage
\bibliographystyle{abbrv}
\bibliography{references}

\end{document}